\documentclass[10pt,twocolumn,twoside]{IEEEtran}
\usepackage[all]{xy}
\usepackage{color}
\usepackage{algorithm}
\usepackage{algpseudocode}
\usepackage{subfigure}
\usepackage{bm}
\usepackage{graphicx}
\usepackage{color}
\usepackage{wrapfig}
\usepackage{epsf}
\usepackage{times,mathptm}
\usepackage{url}
\usepackage{amssymb, amsmath, mathtools, amsthm}
\usepackage{nomencl}
\usepackage{comment}
\let\labelindent\relax
\usepackage{enumitem}
\makenomenclature
\newtheorem{theorem}{Theorem}

\newtheorem{corollary}{Corollary}

\definecolor{RED}{rgb}{0.6,0.,0.}
\definecolor{BLUE}{rgb}{0.,0.,0.6}
\definecolor{GREEN}{rgb}{0.,0.6,0.}
\definecolor{MALINA}{rgb}{0.6,0.,0.6}
\definecolor{YELLOW}{rgb}{0.8,0.8,0}
\newcommand{\squeezeup}{\vspace{-1.5 mm}}
\begin{document}
\title{Joint Estimation of Topology \& Injection Statistics in Distribution Grids with Missing Nodes}
\author{\IEEEauthorblockN{Deepjyoti~Deka, Michael~Chertkov, and Scott~Backhaus\\}
\IEEEauthorblockA{Los Alamos National Laboratory, New Mexico, USA}
\thanks{D. Deka, and M. Chertkov are with the Theory Division and the Center for Nonlinear Systems of LANL, Los Alamos, NM 87544. Email: deepjyoti, chertkov @lanl.gov. S. Backhaus is with the A Division of LANL, Los Alamos, NM 87544. Email: backhaus@lanl.gov}
\thanks{This work was supported by U.S. Department of Energy’s Office of Electricity as part of the DOE Grid Modernization Initiative and the Center for Non Linear Studies at Los Alamos National Laboratory.}}
\maketitle
\begin{abstract}
Optimal operation of distribution grid resources relies on accurate estimation of its state and topology. Practical estimation of such quantities is complicated by the limited presence of real-time meters. This paper discusses a theoretical framework to jointly estimate the operational topology and statistics of injections in radial distribution grids under limited availability of nodal voltage measurements. In particular we show that our proposed algorithms are able to provably learn the exact grid topology and injection statistics at all unobserved nodes as long as they are not adjacent. The algorithm design is based on novel ordered trends in voltage magnitude fluctuations at node groups, that are independently of interest for radial physical flow networks. The complexity of the designed algorithms is theoretically analyzed and their performance validated using both linearized and non-linear AC power flow samples in test distribution grids.
\end{abstract}
\section{Introduction}
\label{sec:intro}
Distribution grid is the part of the power grid network from the distribution substation to the loads and end-users. Often the distribution grid is structured as a radial tree with the substation node as root and load buses/buses powered by the substation located elsewhere. This radial topology is constructed by switching on and off breakers from a set of candidate lines. The optimal operation of smart grids depends on the accurate real-time estimate of the operational topology as well as of the statistics of disturbance/variation in consumption at different grid nodes. However such estimation problems are not straightforward due to the low deployment of real-time meters in the distribution grid \cite{hoffman}. In recent years, there has been a growing adoption of certain `nodal' meters on the distribution side. Examples include distribution PMUs, micro-PMUs \cite{micropmu}, FNETs \cite{fnet}. Additionally, some smart devices at end-user nodes are capable of measuring nodal quantities like voltages for their primary operation. In this paper, we study the problem of structure and statistical estimation in distribution grids using such nodal measurements available only at a subset of the grid nodes - remaining nodes being unobserved/`missing'. Moving forward into the regime of higher meter placement, incomplete observability may still be an issue for third-parties due to privacy concerns and encrypted measurements. As the number of possible radial networks that can be constructed from a set of candidate edges can scale exponentially with its size, brute force methods for topology identification and subsequently injection estimation are avoided. Instead we focus on designing \emph{computationally efficient} theoretical learning algorithms for exact recovery despite the presence of \emph{missing nodes} in the grid.

\subsection{Prior Work}
Learning and estimation in power grids and radial distribution grids in particular has attracted significant attention in recent years. The prior work can be distinguished based on methodology used, assumptions and measurements involved. For available line measurements, \cite{ramstanford} uses maximum likelihood tests for estimating the operational topology using cycles basis. For available nodal voltage measurements, \cite{usc, bolognani2013identification, dekapscc, dekairep, ram_loopy} use properties of the graphical model of voltage measurements to identify the operational topology. Similarly, properties of graphical models in dynamical systems that represent swing dynamics in power grids have been used in grid identification in \cite{sauravacc,sauravacm}. \cite{dekatcns,dekasmartgridcomm} use properties of second moments of voltage magnitudes measurements to identify the radial topology through iterative algorithms that build the tree from leaves to the root. For availability of both voltage and injection measurements, \cite{cavraro2018graph,sejunpscc} design algorithms for topology and parameter (line impedance) identification that considers missing nodes. In agnostic data-driven efforts, topology and parameter identification techniques using machine learning techniques have been discussed in \cite{berkeley,arya}.

An important feature of the majority of cited work based on voltage measurement samples is that exact learning algorithms are only provided for cases with sufficient nodal observability (i.e. without missing nodes). In prior work that considers missing nodes \cite{dekatcns, dekaecc}, topology learning algorithms are designed but require historical knowledge of injections statistics at all nodes, including missing nodes. Such estimates may be unreliable or not present in reality. Further, the hidden nodes are assumed to be separated by three or more hops in the operational grid. We relax both these drawbacks in this paper. In a different setting, \cite{cavraro2018graph,sejunpscc} require both injection and voltage samples at the observed nodes. Availability of real-time injection samples may have stronger consequences for end-user privacy \cite{privacy}. In this work, we consider a setting where samples of nodal voltages and statistics of injections (not samples) are available only at the observed nodes, while missing nodes are two or more hops away (ie. non-adjacent). Our algorithms are able to learn the exact grid topology and estimate the injection statistics at the missing nodes.

\subsection{Technical Contribution}
We consider estimation in partially observed radial grids using time-stamped voltage samples and injection statistics collected from a subset of nodes. Operational edges are selected from among a larger set of permissible edges with known impedances. Under the assumption that missing nodes are non-adjacent and have a degree greater than two, we present learning algorithms to estimate the operational grid topology and estimate the injection statistics at the missing nodes. Based on a linearized AC power flow model \cite{dekatcns,89BWa,89BWb,bolognani2016existence}, we determine relations (equalities and inequalities) between second moments of voltages at groups of two and three nodes that enable guaranteed estimation. We first consider the case with no missing nodes and provide the theoretical sample and computational complexity of a spanning tree based Algorithm \ref{alg:1} originally presented in \cite{dekaecc} that uses only voltage magnitude samples at all nodes to learn the operational topology. We demonstrate through simulations that this algorithm has improved performance at low sample sizes over prior work \cite{dekatcns}. Further we discuss theoretical limitations of Algorithm \ref{alg:1} when missing nodes exist in the system. Next we consider the case where missing nodes are three hops away and present Algorithm \ref{alg:2}, which incorporates additional checks to identify the missing nodes and estimate their injection statistics. Finally we present Algorithm \ref{alg:3} that is able to learn topology and statistics when hidden node are two or more hops away. Going from three to two hop separation uses clustering based on novel monotonic properties of voltages at three nodes. We show the polynomial computational complexity of the designed algorithms, and validate the algorithms on a test distribution system with non-linear AC power flow samples simulated through Matpower \cite{matpower}.

This work is the journal version of a preliminary conference paper \cite{dekasmartgridcomm18} which described Algorithm \ref{alg:3}. This work includes extended proofs of theorems, new Algorithm \ref{alg:2}, as well as theoretical results on sample and computational complexity. Further, unlike \cite{dekasmartgridcomm18}, we present realistic simulation results that enhance the practicality of our proposed algorithms.

The rest of the paper is organized as follows. Section \ref{sec:structure} introduces structural and power flow variables and models used in the remaining sections. Section \ref{sec:trends} describes relations (equalities and inequalities) of second moments of nodal voltage magnitudes that are used in algorithm design. The first learning algorithm for grids with no missing nodes is presented in Section \ref{sec:trends}, along with the analysis of its sample and computational complexity. The second and third learning algorithms for grids with missing nodes are given in Section \ref{sec:learningmissing1} along with derivation of voltage properties that enable their design. Numerical simulations of our learning algorithms on test radial networks are presented in Section \ref{sec:experiments}. Finally, Section \ref{sec:conclusions} contains conclusions and discussion of future work.

\section{Distribution Grid: Structure and Power Flows}
\label{sec:structure}
\begin{figure}[!bt]
\centering
\includegraphics[width=0.21\textwidth]{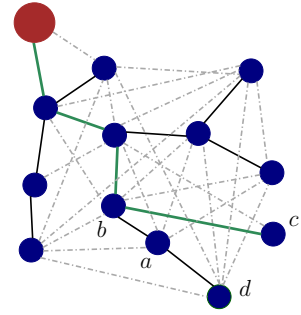}
\squeezeup\caption{Distribution grid with $1$ substation (large red node). The operational lines are solid, and non-operational lines are dashed grey. Nodes $b$ and $a$ are parent-child pair, while $b$ and $d$ are grandparent-grandchild. Node $c$ and $a$ are siblings. Nodes $c,d$ are leaves. The green edges represent path $\mathcal{P}^c$. The descendant set of node $a$ is $\mathcal{D}^a = \{a,d\}$.}\label{fig:picHinv}
\end{figure}
\textbf{Structure}: We represent radial distribution grid by the graph ${\mathcal T}=({\mathcal V},{\mathcal E})$, where ${\mathcal V}$ is the set of buses/nodes of the graph and ${\mathcal E}$ is the set of edges. We denote nodes by alphabets ($a$, $b$,...) and the edge connecting nodes $a$ and $b$ by $(ab)$. The root node of the tree represents a substation and is assumed to have a degree one. This is done for ease of notation as each sub-network emerging from the substation can be separately identified as later discussed. The edge set $\mathcal{E}$ is determined by operational lines (closed) in a set of candidate permissible edges $\mathcal{E}_{full}$. We seek to identify the set of operational edges $\mathcal{E}$ given the set of candidate edges $\mathcal{E}_{full}$. In radial grid $\mathcal{T}$, we denote the unique set of edges that connect a node $a$ to the root node by \emph{path} $\mathcal{P}^a$. We call a node $a$ to be a \textit{descendant} of another node $b$ if $\mathcal{P}^b\subset \mathcal{P}^a$ (i.e. the path from $a$ to root passes through $b$). $\mathcal{D}^a$ is used to denote the set of descendants of $a$. We include node $a$ in $\mathcal{D}^a$ by definition. If $a \in \mathcal{D}^b$ and $(ab) \in \mathcal{E}$, then $b$ and $a$ are termed \emph{parent} and \emph{child} nodes respectively. A parent of a parent is termed \emph{grandparent}. Two nodes that share the same parent are termed \emph{siblings}. Finally terminal nodes that do not have a child are termed \emph{leaves}. An illustrative example of a radial grid with operational edges selected from a candidate set is shown in Fig.~\ref{fig:picHinv} along with the graph-theoretic notations defined. Next we describe the power flow model used in this paper.

\textbf{Power Flow (PF) Model}: Each line $(ab)$ (either operational or open) is associated with a complex impedance $z_{ab}=r_{ab}+\hat{i}x_{ab}$ ($\hat{i}^2=-1$), where $r_{ab}>0$ and $x_{ab}>0$ denote line resistance and reactance respectively. Let the real and reactive injections at node $a$ be denoted by $p_a$ and $q_a$ respectively. Kirchhoff's law relates the complex AC injection at $a$ by the following power flow equation termed AC-PF (Alternating Current Power Flow).
\begin{align}
p_a+\hat{i} q_a = \underset{b:(ab)\in{\mathcal E}}{\sum}\frac{v_a^2-v_a v_b\exp(\hat{i}\theta_a-\hat{i}\theta_b)}{z_{ab}^*}.\label{P-complex1}
\end{align}
Here real valued scalars, $v_a$, $\theta_a$ are the voltage magnitude and phase respectively at node $a$. Under normal operating conditions, small deviations in voltage magnitude from nominal value ($1 p.u.$) at each node and small phase differences between neighboring nodes can be assumed and the following linearized power flow model is derived by ignoring second order terms: \cite{dekatcns,bolognani2016existence}:
\begin{align}
p_a&=\underset{b:(ab)\in{\mathcal E}}{\sum}\left(r_{ab}(v_a-v_b)+ x_{ab}(\theta_a-\theta_b)\right)/\left({x_{ab}^2+r_{ab}^2}\right),
\label{LC-PF_p}\\
q_a&=\underset{b:(ab)\in{\mathcal E}}{\sum}\left(x_{ab}(v_a-v_b) -r_{ab}(\theta_a-\theta_b)\right)/\left({x_{ab}^2+r_{ab}^2}\right)
\label{LC-PF_q}
\end{align}
We term Eqs.~(\ref{LC-PF_p},\ref{LC-PF_q}) as LC-PF (Linear Coupled Power Flow). Note that the active and reactive injections in LC-PF are linear functions of differences in nodal voltage magnitudes and phases. Thus the equations are satisfied if the voltage and phase at all buses are measured relative to some reference bus. Here we consider the substation root node as reference bus with magnitude $1 p.u.$ and phase $0$. Further, summing each equation over all nodes gives $0$. Thereby LC-PF is lossless. Without a loss of generality, we can thus restrict LC-PF analysis to a reduced system without the reference node. This is similar to work in similar lossless models as LinDistFlow \cite{89BWa} or DC power flow. The reduced system is in fact invertible and enables us to express voltages in terms of injections as noted below:
\begin{align}
v = H^{-1}_{1/r}p + H^{-1}_{1/x}q~~ \theta = H^{-1}_{1/x}p - H^{-1}_{1/r}q \label{LC_PF}
\end{align}
Abusing notation, we use $v,\theta, p, q$ to denote the vector of voltage magnitude, phase, active and reactive injections respectively at the non-reference buses in the reduced system. The derivation uses basic matrix inversion. $H_{1/r}$ and $H_{1/x}$ denote the full-rank \emph{reduced weighted Laplacian matrices} for tree $\mathcal T$ where reciprocal of resistances ($1/r$) and reactances ($1/x$) are used respectively as edge weights. The reduction is achieved by removing the row and column corresponding to the reference bus in the original weighted Laplacian matrix. Simulation results on the similarity of LC-PF with non-linear AC power flow generated voltages are described in Section \ref{sec:experiments} for test cases.

For a random vector $X$. we use $\mu_{X} = \mathbb{E}[X]$ to denote its mean. For two random vectors $X$ and $Y$, the centered covariance matrix is denoted by $\Omega_{XY} = \mathbb{E}[(X-\mu_{X})(Y-\mu_{Y})^T]$. If $X=Y$, we denote the covariance matrix by $\Omega_X$.

As LC-PF is linear, we relate the means and covariances of voltage magnitudes and phases with that of the active and reactive injections.

\begin{subequations}\label{moments}
\footnotesize
\begin{align}
\mu_v &= H^{-1}_{1/r}\mu_p + H^{-1}_{1/x}\mu_q,~~ \mu_{\theta} = H^{-1}_{1/x}\mu_p - H^{-1}_{1/r}\mu_q \label{means}\\
\Omega_{v} &= H^{-1}_{1/r}\Omega_{p}H^{-1}_{1/r} + H^{-1}_{1/x}\Omega_q H^{-1}_{1/x}+H^{-1}_{1/r}\Omega_{pq}H^{-1}_{1/x} +H^{-1}_{1/x}\Omega_{qp}H^{-1}_{1/r}\label{volcovar1}\\
\Omega_{\theta} &= H^{-1}_{1/x}\Omega_{p}H^{-1}_{1/x} + H^{-1}_{1/r}\Omega_q H^{-1}_{1/r}-H^{-1}_{1/x}\Omega_{pq}H^{-1}_{1/r} -H^{-1}_{1/r}\Omega_{qp}H^{-1}_{1/x}\label{phasecovar1}\\
\Omega_{v\theta} &= H^{-1}_{1/r}\Omega_{p}H^{-1}_{1/x}-H^{-1}_{1/x}\Omega_{q}H^{-1}_{1/r} - H^{-1}_{1/r}\Omega_{pq} H^{-1}_{1/r}+H^{-1}_{1/x}\Omega_{qp}H^{-1}_{1/x} \label{volphasecovar1}
\end{align}
\end{subequations}
We look at functions of the covariance matrices in the next section and prove equality and inequality results that enable topology and statistical estimation.

\section{Properties of Voltage Second Moments}
\label{sec:trends}
At the onset, we make the following assumption regarding statistics of nodal injections at the grid nodes.

\textbf{Assumption $1$:} Active and reactive injections at different nodes are not correlated, while at the same node are non-negatively correlated. Mathematically, $\forall a,b$ non-substation nodes
\begin{align}
\Omega_{qp}(a,a) \geq 0,~\Omega_p(a,b) = \Omega_q(a,b)= \Omega_{qp}(a,b) = 0 \nonumber
\end{align}
This assumption, similar to ones in \cite{bolognani2013identification,dekapscc,ram_loopy,dekatcns}, is motivated by the fact that at short time-scales, injection fluctuations are the result of loads changes that are independent across nodes. Note that fluctuations at the same node may be aligned. Under Assumption $1$, we analyze second moments of voltages in radial grid $\mathcal{T}$. First we mention two structural results for inverse weighted Laplacian matrices that are true for radial networks, mentioned in \cite{dekatcns}.
\begin{enumerate}[leftmargin=*]
 \item For nodes $a$ and $b$ in tree $\mathcal T$,
 \begin{align}
 H_{1/r}^{-1}(a,b)= \sum_{(cd) \in {\mathcal P}^a\bigcap {\mathcal P}^b} r_{cd} \label{Hrxinv}
 \end{align}
 \item For parent node $b$ and its child $a$,
 \begin{align}
{\huge H}_{1/r}^{-1}(a,c)-{\huge H}_{1/r}^{-1}(b,c) =\begin{cases}r_{ab} & \quad\text{if node $c \in \mathcal{D}^a$}\\
0 & \quad\text{otherwise,} \end{cases} \label{Hdiff}
\end{align}
\end{enumerate}
Note that ${\mathcal P}^a\bigcap {\mathcal P}^b$ denotes the common edges on paths from nodes $a$ and $b$ to the root. The first result follows from structure of inverse reduced incidence matrices in trees. The second result follows from the first result as for parent-child pair $b,a$ and node $c \notin \mathcal{D}^a$, ${\mathcal P}^a\bigcap {\mathcal P}^c$ and ${\mathcal P}^b\bigcap {\mathcal P}^c$ are identical.

We now consider a specific non-negative function of two nodes $\phi_{ab} =\mathbb{E}[(v_a - \mu_{v_a})-(v_b-\mu_{v_b})]^2 $, which represents the variance of the difference in voltage magnitudes at nodes $a$ and $b$. Using Eq.~(\ref{volcovar1}), $\phi_{ab}$ can be expanded in terms of covariances at nodal injections as

{\footnotesize
\begin{align}
&\phi_{ab} = \Omega_{v}(a,a) + \Omega_{v}(b,b)- 2\Omega_{v}(a,b) \nonumber\\
&= \smashoperator[lr]{\sum_{d \in {\mathcal T}}}(H^{-1}_{1/r}(a,d)- H^{-1}_{1/r}(b,d))^2\Omega_p(d,d) +(H^{-1}_{1/x}(a,d)- H^{-1}_{1/x}(b,d))^2 \Omega_q(d,d) \nonumber\\
&~~+2\left(H^{-1}_{1/r}(a,d)- H^{-1}_{1/r}(b,d)\right)\left(H^{-1}_{1/x}(a,d)- H^{-1}_{1/x}(b,d)\right)\Omega_{pq}(d,d) \label{usediff_1}
\end{align}}
\begin{figure}[bt]
\centering
\hspace*{\fill}
\includegraphics[width=0.3\textwidth]{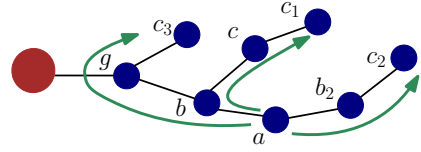}
\squeezeup
\hspace*{\fill}
\caption{Distribution grid tree for Theorem \ref{Theoremcases} illustration. Here $\phi_{ac} = \phi_{ab}+\phi_{bc}$, and $\phi_{ac_1} = \phi_{ab}+\phi_{bc_1}$, while $\phi_{ac_2} >\phi_{ab_2}+\phi_{b_2c_2}$, and $\phi_{ac_3}>\phi_{ab}+\phi_{bc_3}$.}
\label{fig:item1}
\end{figure}
The following result shows increasing trends in $\phi_{ab}$ along paths in the radial grid.
\begin{theorem} \label{Theoremcases}
For three nodes $a \neq b \neq c$ in tree $\mathcal{T}$, let the path from $a$ to $c$ pass through node $b$ in tree $\mathcal{T}$. Then
\begin{enumerate}[leftmargin =*]
 \item $\phi_{ab} + \phi_{bc} = \phi_{ac} \text{~~if~~} \mathcal{P}^c\bigcap \mathcal{P}^a = \mathcal{P}^b$
 \item $\phi_{ab} + \phi_{bc} < \phi_{ac} \text{~~if~~} \mathcal{P}^c\bigcap \mathcal{P}^a \subset \mathcal{P}^b$
\end{enumerate}
\end{theorem}
The proof, originally presented in the conference paper \cite{dekaecc}, is provided in Appendix \ref{sec:proof1} for completion and use in subsequent theorems. Theorem \ref{Theoremcases} states that $\phi$ computed across any path in $\mathcal{T}$ is at least as large as the sum computed across its non-overlapping sub-paths as shown in Fig.~\ref{fig:item1}. The following theorem from \cite{dekaecc} uses this result to estimate the operational tree from the set of permissible edges $\mathcal{E}_{full}$.
\begin{theorem}\label{main}
Let each permissible edge $(ab)$ in ${\mathcal E}_{full}$ be given weight $\phi_{ab} = \mathbb{E}[(v_a-\mu_{v_a}) -(v_b-\mu_{v_b})]^2$. The operational edge set ${\mathcal E}$ is given by the minimum weight spanning tree in set ${\mathcal E}_{full}$.
\end{theorem}
Theorem \ref{main} states that the exact topology of the grid can be computed using just the voltage magnitude measurements at all grid nodes. No additional information related to injection statistics are needed. If voltage phase angles are also available, the injection statistics at all nodes can be computed by inverting Eqs.~(\ref{moments}) or iteratively from leaves to the root using Eq.~(\ref{first}) described later. The steps in topology and injection statistics estimation are listed in Algorithm \ref{alg:1}, originally presented as Algorithm \ref{alg:1} in \cite{dekaecc}.
\begin{algorithm}
\caption{Learning without missing nodes}\label{alg:1}
\textbf{Input:} Voltage observations $v$, $\theta$ at all nodes, set of permissible edges ${\mathcal E}_{full}$ with line impedances.\\
\textbf{Output:} Operational edges $\mathcal{E}$, injection covariances $\Omega_p, \Omega_q, \Omega_{pq}$ at all nodes
\begin{algorithmic}[1]
\State $\forall (ab) \in {\mathcal E}_{full}$, compute $\phi_{ab}=\mathbb{E}[(v_a - \mu_{v_a})-(v_b-\mu_{v_b})]^2$
\State Find min. spanning tree from $\mathcal E$ with $\phi_{ab}$ as edge weights.
\State ${\mathcal E} \gets $ {edges in spanning tree}
\State Compute $\Omega_p, \Omega_q, \Omega_{pq}$ using Eqs.~(\ref{moments}).
\end{algorithmic}
\end{algorithm}

\textbf{Computational Complexity:} For set $\mathcal{E}_{full}$, minimum spanning tree can be found using Kruskal's Algorithm \cite{kruskal1956shortest,Cormen2001} in $O(|{\mathcal E}|\log|{\mathcal E}|)$ operations. In the worst case, where all node pairs are permissible edges, the complexity scales as $O(|{\mathcal V}|^2\log|{\mathcal V}|)$.
The next result presents the number of voltage samples necessary for accurate recovery using empirical estimates of $\phi$.
\begin{theorem} \label{theorem:sample_complexity}
For radial grid $\mathcal{T}$ with node set ${\mathcal V}$ and depth $d$, assume line impedances are bounded by non-zero values and nodal injections to be zero-mean Gaussians with bounded variance. For $0 < \eta<1$, if the number of nodal voltage magnitude samples $n$ is greater than $Cd^4|\mathcal V|^2\log(|\mathcal V|/\eta)$ for some constant $C$, then Algorithm \ref{alg:1} recovers the true topology with probability $1-\eta$.
\end{theorem}
The proof is given in Appendix \ref{sec:proofsample}. In a realistic grid, all nodes may not observed. Naive application of Algorithm \ref{alg:1} can lead to errors in topology estimation as noted in the following result.
\begin{figure}[!bt]
\centering
\hspace*{\fill}
\subfigure[]{\includegraphics[width=0.1\textwidth]{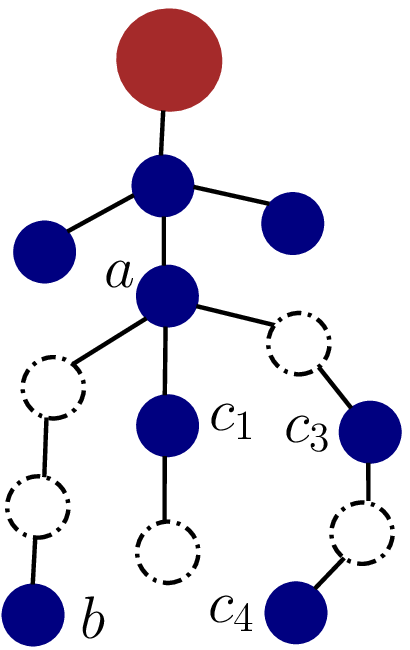}\label{fig:missing_a}}\hfill
\subfigure[]{\includegraphics[width=0.1\textwidth]{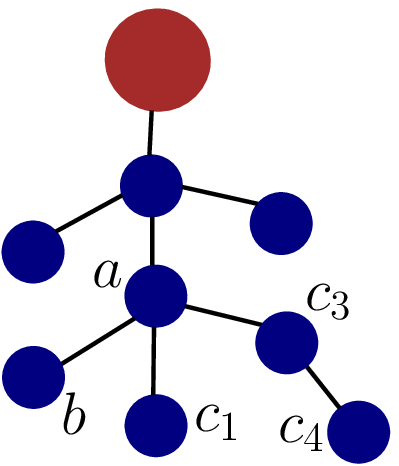}\label{fig:missing_b}}\hfill
\hspace*{\fill}
\squeezeup
\caption{(a) Distribution grid tree ${\mathcal T}$ with unobserved nodes of degree less than $3$.(b) Output of applying Algorithm \ref{alg:1}\label{fig:missing0}}
%\vspace{-3mm}
\end{figure}
\begin{theorem}\label{thm:hiddendegree1}
Consider missing nodes of degree at most $2$ in grid tree $\mathcal T$. Algorithm \ref{alg:1} using observed node voltages creates a tree ${\mathcal T}_{\mathcal M}$ where observed nodes in $\mathcal T$ separated by missing nodes are connected by spurious edges, while rest of the true edges are identified.
\end{theorem}
\begin{proof}
Using theorem \ref{Theoremcases}, it is clear that observed neighbors in $\mathcal{T}$ will be neighbors in ${\mathcal T}_{\mathcal M}$. As missing nodes have maximum degree $2$, there is at most a line sub-graph of connected hidden nodes with observed nodes at either end (see Fig.~\ref{fig:missing0}). These observed nodes have the lowest $\phi$ among nodes separated by the hidden nodes, hence edges between them appear in the spanning tree ${\mathcal T}_{\mathcal M}$.
\end{proof}
The following result follows immediately from Theorem \ref{thm:hiddendegree1}.
\begin{corollary}
If grid $\mathcal{T}$ of $|\mathcal V|$ nodes has $k$ non-adjacent missing nodes, each of degree $2$, then Algorithm \ref{alg:1} produces a tree ${\mathcal T}_{\mathcal M}$ of $|\mathcal V|-k$ nodes with $k$ spurious edges not present in $\mathcal T$, and does not have $2k$ missing edges from $\mathcal T$.
\end{corollary}
The next section presents additional results on nodal voltages and discusses tractable learning in the presence of missing nodes.
\section{Learning with missing nodes}
\label{sec:learningmissing1}
We consider voltage measurements and knowledge of injection statistics at the observed nodes while the missing nodes are unobserved. First we consider the setting where missing nodes are separated by greater than two hops.
\subsection{Missing nodes separated by three or more hops}
\label{sec:hidden_3}
Let the set of observed nodes be $\mathcal O$, i.e., where voltage measurements and injection covariances are known. We consider arbitrary placement of unobserved node set $\mathcal M$ with no measurements or historical data under the following restriction in this section.

\textbf{Assumption $2$:} All missing nodes have a degree greater than $2$ and are separated by greater than two hops in the grid tree $\mathcal T$.

The degree assumption ensures uniqueness of topology reconstruction. In particular, if hidden nodes of degree $2$ are adjacent, one can combine them into a single hidden node by Kron reduction (similar to Theorem \ref{thm:hiddendegree1}) while maintaining consistency with available measurements. This prevents unique reconstruction. Note that under Assumption $2$, no hidden node is a leaf.

Consider a tree $\mathcal T$ where missing node set $\mathcal M$ satisfies Assumption $2$. Let the minimum spanning tree ${\mathcal T}_{\mathcal M}$ between observed nodes $\mathcal{O}$ be constructed using Algorithm \ref{alg:1} with $\phi$'s as edge weights. Consider the case shown in Fig.~\ref{fig:missing1} with missing node $b$. By Assumption $2$, all nodes within two hops of $b$ are observed. Hence its parent $a$, children node set ${\mathcal C}_b = \{c_1, c_2, c_3, c_4\}$ are observed. Also all neighbors of $a$ and ${\mathcal C}_b$ except $b$ are observed. By Theorem \ref{Theoremcases}, all edges between $a$ and non-descendants of $b$ in $\mathcal{T}_{\mathcal M}$ are true edges, while observed descendants of $b$ are connected to the rest of ${\mathcal T}_{\mathcal M}$ through false\footnote{non-existent edges} edges between ${\mathcal C}_b$ and $a$. The following theorem gives possible configurations between ${\mathcal C}_b$ and $a$ in ${\mathcal T}_{\mathcal M}$.
\begin{figure}[!bt]
\centering
\hspace*{\fill}
\subfigure[]{\includegraphics[width=0.15\textwidth]{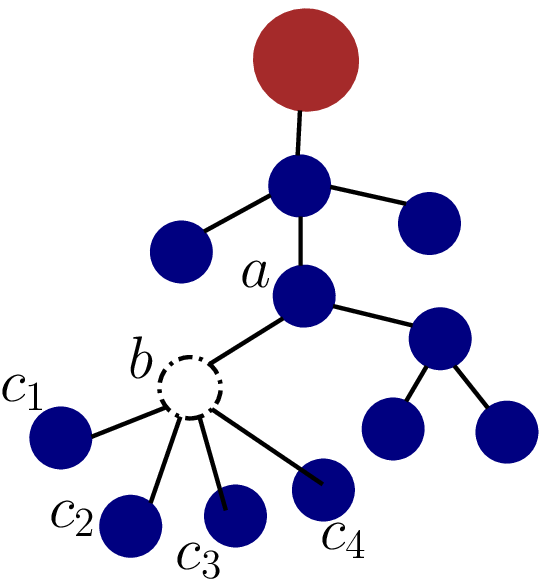}\label{fig:missing1}}\hfill
\subfigure[]{\includegraphics[width=0.15\textwidth]{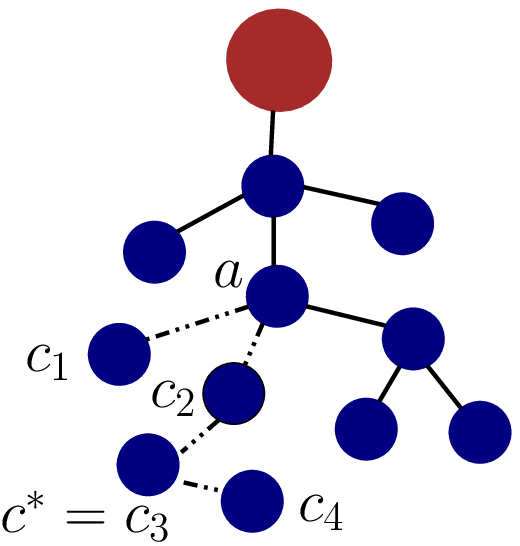}\label{fig:missing2}}\hfill
\hspace*{\fill}
\squeezeup
\caption{(a)Distribution grid tree ${\mathcal T}$ with unobserved node $b$. Node $a$ is $b$'s parent while nodes $c_1,c_2,c_3,c_4$ are its children. (b) Possible configuration of spanning tree ${\mathcal T}_{\mathcal M}$ of observed nodes as per Theorem \ref{permissiblecases}\label{fig:missing}}
%\vspace{-3mm}
\end{figure}
\begin{theorem}\label{permissiblecases}
For missing node $b$ in $\mathcal T$ with observed parent $a$ and observed children node set ${\mathcal C}_b$, let $\arg\min\limits_{c_i \in {\mathcal C}_b} \phi_{bc_i} = c^*$. Then
\begin{itemize}[leftmargin=*]
\item No edge $(c_ic_j)$ between children $c_i, c_j \neq c^*$ exists in ${\mathcal T}_{\mathcal M}$.
\item Nodes in set ${\mathcal C}_b^1= \{c_i\in {\mathcal C}_b, \phi_{ac_i} < \phi_{c^*c_i}\}$ are connected to node $a$, those in $\mathcal{C}_b-{\mathcal C}_b^1-\{c^*\}$ are connected to $c^*$.
\end{itemize}
\end{theorem}
\begin{proof}
Consider node pair $c_i,c_j\neq c^*$ in $\mathcal{C}_b$. Using Eq.~(\ref{second}) $\phi_{c_ic_j} = \phi_{bc_i}+ \phi_{bc_j} <\phi_{bc_i}+ \phi_{bc^*} = \phi_{c_ic^*}$. Thus, any possible edge between nodes in $\mathcal{C}_b$ includes node $c^*$. The edges for each node in sets ${\mathcal C}_b^1$ and $\mathcal{C}_b -{\mathcal C}_b^1$ follows by definition of min-weight spanning tree.
\end{proof}
Note that one of the sets ${\mathcal C}^1$ or ${\mathcal C}^2$ may be empty. It is worth mentioning that node $c^*$ can be connected to some node $c\dag \in {\mathcal C}^1$ instead of directed to $a$ if $\phi_{ac\dag} < \phi_{c^*c\dag} <\phi_{ac^*}$ holds. Theorem \ref{permissiblecases} thus suggests that if Algorithm \ref{alg:1} outputs $\mathcal{T}_{\mathcal{M}}$ between observed nodes, it may include false edges between an observed node to either its siblings (for missing parent), or to its grandchildren (for a single missing child). This is depicted in Fig.~\ref{fig:missing2}. In particular, two sibling nodes with missing parent in $\mathcal T$ may be as far as four hops away in ${\mathcal T}_{\mathcal M}$. Note that unlike the case for missing nodes of degree $2$ (see Theorem \ref{thm:hiddendegree1}), here multiple configurations may be possible.

To estimate the operational edges, locate the missing nodes and estimate their injections statistics, we require additional properties of $\phi$ that make learning tractable. First we prove equality relations for $\phi$ computed for parent-child nodes and parent-grandchildren nodes.

\begin{theorem} \label{Theoremcases2}
In $\mathcal T$, the following statements hold:\\\\
%\begin{enumerate}[leftmargin=*]
$1$. If node $b$ is the parent of nodes $a$ and $c$ (see Fig.~\ref{fig:item1})
\begin{subequations}
\footnotesize
\begin{align}
\phi_{ab} &= \smashoperator[lr]{\sum_{d \in \mathcal{D}^a}}r_{ab}^2\Omega_p(d,d)+x_{ab}^2 \Omega_q(d,d)+2r_{ab}x_{ab}\Omega_{pq}(d,d)\label{first}\\
\phi_{ac} &= \smashoperator[lr]{\sum_{d \in \mathcal{D}^a}}r_{ab}^2\Omega_p(d,d)\nonumber+x_{ab}^2 \Omega_q(d,d)+2r_{ab}x_{ab}\Omega_{pq}(d,d)\nonumber\\
&+ \smashoperator[lr]{\sum_{d \in \mathcal{D}^c}}r_{bc}^2\Omega_p(d,d)+x_{bc}^2 \Omega_q(d,d)+2r_{bc}x_{bc}\Omega_{pq}(d,d)\label{second}
\end{align}
\end{subequations}
$2$. If node $g$ is the parent of node $b$ and grandparent of nodes $a$ and $c$ (see Fig.~\ref{fig:item1}),
\begin{subequations}
\footnotesize
\begin{flalign}
&\phi_{ag}-\phi_{cg} = \smashoperator[lr]{\sum_{d \in \mathcal{D}^a}}\Omega_p(d,d)(r_{ab}^2+ 2r_{ab}r_{bg}) + \Omega_q(d,d)(x_{ab}^2+ 2x_{ab}x_{bg})\nonumber\\ &+2\Omega_{pq}(d,d)(r_{ab}x_{ab}+r_{bg}x_{ab}+r_{ab}x_{bg})-
\smashoperator[lr]{\sum_{d \in \mathcal{D}^c}}\Omega_p(d,d)(r_{cb}^2+ 2r_{cb}r_{bg}) \nonumber\\
&-\Omega_q(d,d)(x_{cb}^2+ 2x_{cb}x_{bg}) -2\Omega_{pq}(d,d)(r_{cb}x_{cb}+r_{bg}x_{cb}+r_{cb}x_{bg})\text{~and,}\label{third}\end{flalign}
\begin{flalign}
&\phi_{ag} = \smashoperator[lr]{\sum_{d \in \mathcal{D}^a}}\Omega_p(d,d)(r_{ab}+ r_{bg})^2 +2\Omega_{pq}(d,d)(r_{ab}+r_{bg})(x_{ab}+x_{bg})\label{fourth}\\
+ &\Omega_q(d,d)(x_{ab}+ x_{bg})^2+ \smashoperator[lr]{\sum_{d \in \mathcal{D}^b-\mathcal{D}^a}}\Omega_p(d,d)r_{bg}^2+\Omega_q(d,d)x_{bg}^2+2\Omega_{pq}(d,d)r_{bg}x_{bg}\nonumber\\
&\phi^{\theta}_{ag} = \smashoperator[lr]{\sum_{d \in \mathcal{D}^a}}\Omega_p(d,d)(x_{ab}+ x_{bg})^2 -2\Omega_{pq}(d,d)(r_{ab}+r_{bg})(x_{ab}+x_{bg})\label{fifth}\\
+& \Omega_q(d,d)(r_{ab}+ r_{bg})^2+ \smashoperator[lr]{\sum_{d \in \mathcal{D}^b-\mathcal{D}^a}}\Omega_p(d,d)x_{bg}^2+\Omega_q(d,d)r_{bg}^2-2\Omega_{pq}(d,d)r_{bg}x_{bg}\nonumber\\
&\phi^{v\theta}_{ag}= \smashoperator[lr]{\sum_{d \in \mathcal{D}^a}}(\Omega_p(d,d)-\Omega_q(d,d))(r_{ab}+r_{bg})(x_{ab}+x_{bg})+\Omega_{pq}(d,d)(x_{ab}+x_{bg})^2\label{sixth}\\
-&\Omega_{pq}(d,d)(r_{ab}+r_{bg})^2+\smashoperator[lr]{\sum_{d \in \mathcal{D}^b-\mathcal{D}^a}}(\Omega_p(d,d)-\Omega_q(d,d))r_{bg}x_{bg}+\Omega_{pq}(d,d)(x^2_{bg}-r^2_{bg})\nonumber
\end{flalign}
\end{subequations}
%\end{enumerate}
where $\phi^{v\theta}_{ab} =\mathbb{E}[(v_a -\mu_{v_a}-v_b+\mu_{v_b})(\theta_a - \mu_{\theta_a}-\theta_b+\mu_{\theta_b})]$ and $\phi^{\theta}_{ab} =\mathbb{E}[(\theta_a - \mu_{\theta_a})-(\theta_b-\mu_{\theta_b})]^2$.
\end{theorem}
The derivation of statement $1$ in Theorem \ref{Theoremcases2} follows the derivation of the first statement in Theorem \ref{Theoremcases} for parent-child pairs. The second statement is proven by expanding $\phi$, $\phi^{\theta}$ and $\phi^{v\theta}$ for grandchildren-grandparent pairs using Eq.~(\ref{usediff_1}) and Eq.~(\ref{Hdiff}). We mention key takeaways from Theorem \ref{Theoremcases2} that enable verification of relative nodal positions in tree $\mathcal T$ and estimate injection statistics.
\begin{enumerate}[leftmargin =*]
 \item If all descendants of nodes $a$ are known then Eq.~(\ref{first}) can be used to verify its parent.
 \item If $a$ and $c$ are known siblings and their descendants are known, then Eq.~(\ref{second}) can be used to search for their parent $b$ among possible edges in $\mathcal{E}_{full}$.
 \item If $a$ and $c$ are siblings with known grandparent $g$ and descendant sets $\mathcal{D}^a,\mathcal{D}^c$, Eq.~(\ref{third}) can be used to search for $a$ and $c$'s parent.
 \item If the injections at all descendants of node $b$ is known and its parent is verified as $g$, Eqs.~(\ref{fourth}-\ref{sixth}) can be used to determine its injection statistics.
\end{enumerate}
Note that identification of parents as listed above (takeaways $2,3$) involves a linear search over the set of permissible edges and hence is not computationally intensive. In the final takeaway, the estimation of $b$'s injection statistics ($\Omega_p(b,b),\Omega_q(b,b),\Omega_{pq}(b,b)$) involves solving three linear equations with three unknowns if all its descendants are known.

These results are used next to jointly estimate topology and injection statistics in the presence of missing nodes. The overall steps in the learning procedure are listed in Algorithm \ref{alg:2}.
\begin{algorithm}
\caption{Learning with Hidden Nodes separated by more than $2$ hops}\label{alg:2}
\textbf{Input:} Voltage observations $v$, $\theta$, and injection covariances $\Omega_p, \Omega_q, \Omega_{pq}$ at available node set $\mathcal O$, hidden node set ${\mathcal M}$, set of permissible edges ${\mathcal E}_{full}$ with line impedances, thresholds $\tau_1,\tau_2$.\\
\textbf{Output:} Operational edges ${\mathcal E}$, $\Omega_p, \Omega_q, \Omega_{pq}$ at set $\mathcal{M}$
\begin{algorithmic}[1]
\State $\forall$ nodes $a,b \in {\mathcal O}$, compute $\phi_{ab}$ and find minimum weight spanning tree ${\mathcal T}_{\mathcal{M}}$ with $\phi_{ab}$ as edge weights. \label{span2}
\State Sort nodes in ${\mathcal T}_{\mathcal{M}}$ in decreasing order of their depths and mark them as unexplored.
\While {$|{\mathcal M}| >0$ OR no node unexplored}
\State Select unexplored node $a$ with parent $p$ at greatest depth with observed children set $\mathcal{C}_a$ and undetermined grandchildren set $G_{a}$ in ${\mathcal T}_{\mathcal{M}}$.
\ForAll{$b \in \mathcal{C}_a$}\label{step_parent21}
\If {$\phi_{ab}$ satisfy Eq.~(\ref{first}) with threshold $\tau_1$}
\State ${\mathcal E} \gets {\mathcal E} \cup \{(ab)\}$, $\mathcal{C}_a\gets \mathcal{C}_a-\{b\}$
\EndIf
\EndFor\label{step_parent22}
\State $\mathcal{C}_a \gets \mathcal{C}_a \bigcup G_a$
\For{$b \in \mathcal{M}$,$|\mathcal{C}_a|\geq 2$} \label{step_grandparent21}
\If{$a$, child $b$, grandchildren in $\mathcal{C}_a$ satisfy Eq.~(\ref{third}) with threshold $\tau_2$}
\State ${\mathcal E}\gets {\mathcal E} \bigcup \{(ba)\}\bigcup\{(bc)\forall c\in \mathcal{C}_a\}$
\State Solve $\Omega_p(b,b), \Omega_q(b,b), \Omega_{pq}(b,b)$ from (\ref{fourth}-\ref{sixth}).
\State $\mathcal{C}_a\gets \{\}$, $\mathcal M\gets \mathcal M-\{b\}$.
\EndIf
\EndFor\label{step_grandparent22}
\If{$|\mathcal{C}_a|\geq 0$}\label{step_sibling21}
\State Disconnect $(ap)$ from $a$'s parent $p$ in $\mathcal{T}_{\mathcal{M}}$. Expand undetermined grandchildren set of $p$, $G_{p} \gets G_{p}\bigcup\mathcal{C}_a\bigcup\{a\}$
\EndIf\label{step_sibling22}
\State Mark $a$ as explored
\EndWhile
\end{algorithmic}
\end{algorithm}

\textbf{Algorithm \ref{alg:2} working:} We first construct the spanning tree $\mathcal{T}_{\mathcal{M}}$ of observed nodes using $\phi$ as edge weights of permissible edges in set $\mathcal{E}_{full}$ (Step \ref{span2}). To determine missing nodes and their injection statistics, we iteratively verify edges starting from leaves to the root in $\mathcal{T}_{\mathcal M}$. This is done as checks at a node depend on injections at its descendants that may be missing. We consider observed non-leaf nodes at the greatest depth in ${\mathcal T}_{\mathcal{M}}$ to iteratively search for hidden nodes with the first iteration involving parents of leaf nodes. We first use Eq.~(\ref{first}) to verify whether each edge is true (Steps \ref{step_parent21}-\ref{step_parent22}). If edges to some set $\mathcal{C}_a$ are not verified, we check if $a$ is their grandparent with some missing parent $b$ using Eq.~(\ref{third}) (Steps \ref{step_grandparent21}-\ref{step_grandparent22}). From Assumption $2$, nodes in $\mathcal{C}_a$ can have one missing parent. If missing parent is identified, its injections are estimated using Eqs.~(\ref{fourth}). If not confirmed, we list $a$ and $\mathcal{C}_a$ as siblings with unknown parent under $a$'s previous parent $p$ (Steps \ref{step_sibling21}-\ref{step_sibling22}). $a$ is marked as explored and the algorithm looks at the next unexplored node.

\textbf{Computational Complexity:} As before, we can compute the spanning tree for observed nodes in $O((|{\mathcal V}|-|{\mathcal M}|)^2\log(|{\mathcal V}|-|{\mathcal M}|))$ in worst case when all edges between observed nodes are permissible. Next we sort the observed nodes in topological order in linear time $O(|{\mathcal V}|-|{\mathcal M}|)$ \cite{Cormen2001}. Checking the parent-child and grandparent-grandchildren relations has complexity $O((|{\mathcal V}|- |{\mathcal M}|)|{\mathcal V}|)$ due to iterations over $O(|{\mathcal V}|- |{\mathcal M}|)$ observed nodes in $\mathcal{T}_{\mathcal{M}}$ with possible search over each child and each missing node. The overall complexity is thus $O(|{\mathcal V}|^2\log |{\mathcal V}|)$ in the worst case.

In the next section, we extend Algorithm \ref{alg:2} to consider cases where missing nodes can be two hops away instead of three.

\subsection{Missing nodes separated by two or more hops}
\label{sec:hidden_2}
Here we consider missing nodes' placement under the following assumption.

\textbf{Assumption $3$:} All missing nodes have a degree greater than $2$ and are not adjacent in the grid tree $\mathcal T$.

Under Assumption $3$, both parent and multiple children of an observed node $a$ may be missing (see Fig. \ref{fig:twohop}). This is unlike Assumption $2$ where only parent or one child of $a$ may be missing. Let $\mathcal{T}_{\mathcal{M}}$ be the spanning tree of observed nodes given by Algorithm \ref{alg:1}. In $\mathcal{T}_{\mathcal{M}}$ under Assumption $3$, $a$ may thus be connected as parent to its siblings (from missing parent), as well as to its grandchildren (from multiple missing children) as depicted in Fig.~\ref{fig:missing_1}. Thus, observed nodes that are four hops away in $\mathcal T$ may be two hops away in $\mathcal{T}_{\mathcal M}$.

To distinguish true siblings and true grandchildren in $\mathcal T$ among false children in the spanning tree of observed nodes, we use additional voltage inequalities at node triplets (groups of three), described next.\squeezeup\squeezeup\squeezeup
\begin{figure}[hbt]
\centering
\hspace*{\fill}
\subfigure[]{\includegraphics[width=0.15\textwidth]{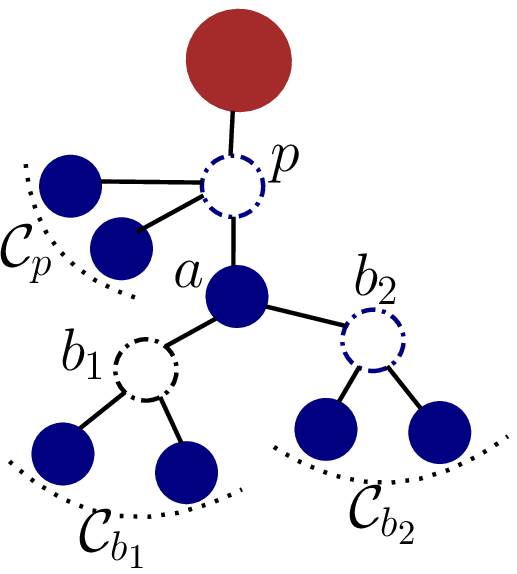}\label{fig:twohop}}\hfill
\subfigure[]{\includegraphics[width=0.18\textwidth]{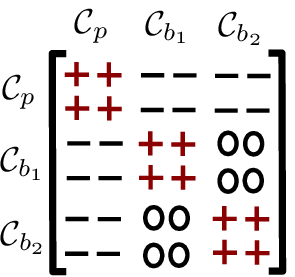}\label{fig:matrix}}
\squeezeup
\hspace*{\fill}
\caption{(a) Node $a$ with parent $p$ and children $b_1,b_2$. Node $a$ has siblings $C_p$, grandchildren $C_{b_1},C_{b_2}$. (b) $[\phi_{k_1a} -\phi_{k_2a}+\phi_{k_1k_2}]$ for $k_1,k_2 \in C_{b_1},C_{b_2},C_p$}
\end{figure}

\begin{theorem} \label{Theoreminequality}
Consider node $a$ in $\mathcal T$ with parent $p$ and children nodes $b_1,b_2$. Let $\mathcal{C}_p$ be set of sibling nodes of $a$ with parent $p$ (see Fig.~\ref{fig:twohop}). Let $\mathcal{C}_{b_1} $ be children nodes of $b_1$ and $\mathcal{C}_{b_2}$ be children of $b_2$. Then the following inequalities hold:
\begin{enumerate}[leftmargin = *]
\item $\phi_{k_1a} +\phi_{k_2a}- \phi_{k_1k_2} > 0 \text{~if~} k_1,k_2 \text{~are siblings in~} \mathcal{C}_{b_1}, \mathcal{C}_{b_2},\text{~or~}\mathcal{C}_p$
\noindent\item $\phi_{k_1a} +\phi_{k_2a}- \phi_{k_1k_2} = 0 \text{~if~} k_1 \in \mathcal{C}_{b_1},k_2 \in \mathcal{C}_{b_2}$
\noindent\item $\phi_{k_1a} +\phi_{k_2a}- \phi_{k_1k_2} < 0 \text{~if~} k_1 \in \mathcal{C}_{b_1} \text{or~} \mathcal{C}_{b_2}$ and $k_2 \in \mathcal{C}_{p}$
\end{enumerate}
\end{theorem}
\begin{proof}
To simplify notation, we consider sets $\mathcal{C}_p = \{c_1,c_2\}, \mathcal{C}_{b_1} = \{c_3,c_4\}, \mathcal{C}_{b_2} = \{c_5,c_6\}$ as shown in Fig.~\ref{fig:twohop}. Using the first result in Theorem \ref{Theoremcases}, we have
\begin{align*}
 &\phi_{c_1a} =\phi_{ap} + \phi_{c_1p},~ \phi_{c_1c_2} = \phi_{c_1p} + \phi_{c_2p},~ \phi_{c_2a} = \phi_{ap} + \phi_{c_2p}\nonumber\\
 \Rightarrow~&\phi_{c_1a} +\phi_{c_2a}- \phi_{c_1c_2} > 0.
\end{align*}
Now consider grandchildren of node $a$ and children of $b_1$. From second result in Theorem \ref{Theoremcases}, we have
\begin{align*}
&\phi_{c_3a} > \phi_{c_3b_1},~ \phi_{c_4a} >\phi_{c_4b_1}\nonumber\\
\Rightarrow~&\phi_{c_3a} +\phi_{c_4a}- \phi_{c_3c_4}>0 ~(\text{as}~ \phi_{c_3c_4} = \phi_{c_3b_1} + \phi_{c_4b_1})\nonumber
\end{align*}
By symmetry it is true for $c_5,c_6\in \mathcal{C}_{b_2}$. This proves the first statement. Statement $2$ follows immediately from the first result in Theorem \ref{Theoremcases}. For Statement $3$, consider the case $k_1 = c_3 \in \mathcal{C}_{b_1}, k_2 = c_1\in \mathcal{C}_p$. We have
\begin{align*}
\phi_{c_3c_1} =\phi_{c_3p}+\phi_{c_1p}> \phi_{c_3a}+\phi_{ap}+\phi_{c_1p} = \phi_{c_3a}+\phi_{c_1p}
\end{align*}
where the inequality follows from the second result in Theorem \ref{Theoremcases}.
\end{proof}
The key result of Theorem \ref{Theoreminequality} is effective depicted in Fig.~\ref{fig:matrix} through the matrix $[\phi_{k_1a} +\phi_{k_2a}- \phi_{k_1k_2}]$ constructed using siblings or grandchildren of node $a$. Note that the positive values in the matrix correspond to siblings of common parent. Hence it can be used to distinguish erroneous children of a node into its siblings and grandchildren in our learning algorithm. The true parent of each grandchildren group can be identified and its injection statistics estimated using Eq.~(\ref{third}) and Eq.~(\ref{fourth}-\ref{sixth}) in Theorem \ref{Theoremcases2}. Next, we design Algorithm \ref{alg:3} to learn the topology and injection statistics with non-adjacent missing nodes.

\begin{algorithm}
\caption{Learning with Hidden Nodes separated by more than $1$ hop}\label{alg:3}
\textbf{Input:} Voltage observations $v$, $\theta$, and injection covariances $\Omega_p, \Omega_q, \Omega_{pq}$ at available node set $\mathcal O$, hidden node set ${\mathcal M}$, set of permissible edges ${\mathcal E}_{full}$ with line impedances, thresholds $\tau_1, \tau_2,\tau_3$\\
\textbf{Output:} Operational edges ${\mathcal E}$, $\Omega_p, \Omega_q, \Omega_{pq}$ at set $\mathcal{M}$
\begin{algorithmic}[1]
\State $\forall$ nodes $a,b \in {\mathcal O}$, compute $\phi_{ab}$ and find minimum weight spanning tree ${\mathcal T}_{\mathcal{M}}$ with $\phi_{ab}$ as edge weights. \label{span3}
\State Sort nodes in ${\mathcal T}_{\mathcal{M}}$ in decreasing order of their depths and mark them as unexplored.
\While {$|{\mathcal M}| >0$ OR no node unexplored}
\State Select in ${\mathcal T}_{\mathcal{M}}$ unexplored node $a$ with parent $p$ at greatest depth with observed children set $\mathcal{C}_a$ and undetermined grandchildren sets $G^i_{a}, i = 1,2..$.
\ForAll{$b \in \mathcal{C}_a$}\label{step_parent31}
\If {$\phi_{ab}$ satisfy Eq.~(\ref{first}) with threshold $\tau_1$}
\State ${\mathcal E} \gets {\mathcal E} \cup \{(ab)\}$, $\mathcal{C}_a\gets \mathcal{C}_a-\{b\}$
\EndIf
\EndFor\label{step_parent_32}
\State Take one grandchild $g_i$ per $G^i_a$ and nodes in $\mathcal{C}_a$ and separate them into grandchildren sets $G^i_a$ and sibling set $\mathcal{S}_a$ by clustering $\phi$ using Theorem \ref{Theoreminequality} with threshold $\tau_3$. Add siblings of each $g_i$ to its separated set. \label{step_grandparent31}
\State Find missing parent of separated grandchildren set $G_i$ using Eq.~(\ref{third}) with threshold $\tau_2$, determine its injection statistics using Eqs.~(\ref{fourth}-\ref{sixth}) and remove it from $\mathcal M$. Add discovered edges to $\mathcal{E}$.\label{step_grandparent32}
\If{$|\mathcal{S}_a|\geq 0$}\label{step_sibling31}
\State Disconnect $(ap)$ from $a$'s parent $p$ in $\mathcal{T}_{\mathcal{M}}$. Form undetermined grandchildren group $G^i_{p}$ with $\mathcal{S}_a$ and $a$
\EndIf\label{step_sibling32}
\State Mark $a$ as explored
\EndWhile
\end{algorithmic}
\end{algorithm}
\textbf{Algorithm \ref{alg:3} working:} The basic working of Algorithm \ref{alg:3} follows a similar logic as Algorithm \ref{alg:2}. The differences exist in Steps (\ref{step_grandparent31}-\ref{step_grandparent32}) where Theorem \ref{Theoreminequality} is used to separate siblings of a current node $a$ from its grandchildren and then to identify its missing children and estimate their injection statistics. For better elucidation, the steps in Algorithm \ref{alg:3} for estimating the grid in Fig.~\ref{fig:twohop} are depicted in Fig.~\ref{fig:missing_example}. Note that the hidden nodes are $p,b_1,b_2$.
\begin{figure}[!ht]
\centering
\hspace*{\fill}
\subfigure[]{\includegraphics[width=0.11\textwidth]{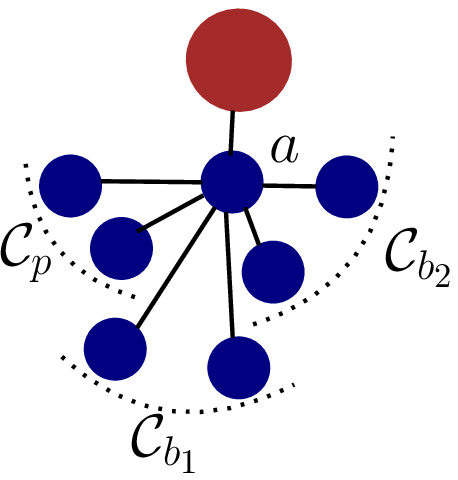}\label{fig:missing_1}}
\subfigure[]{\includegraphics[width=0.11\textwidth]{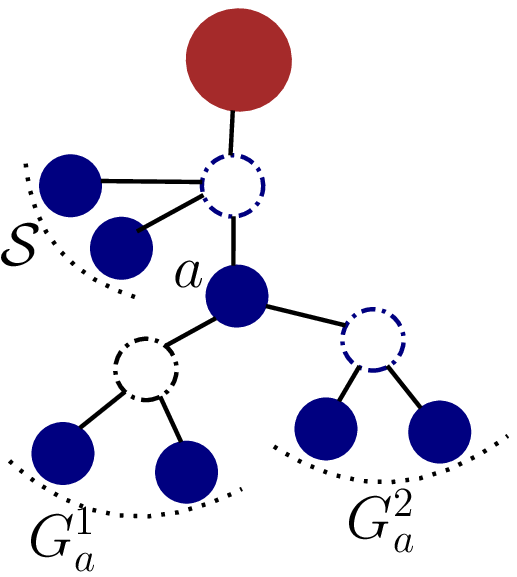}\label{fig:missing_1a}}
\subfigure[]{\includegraphics[width=0.11\textwidth]{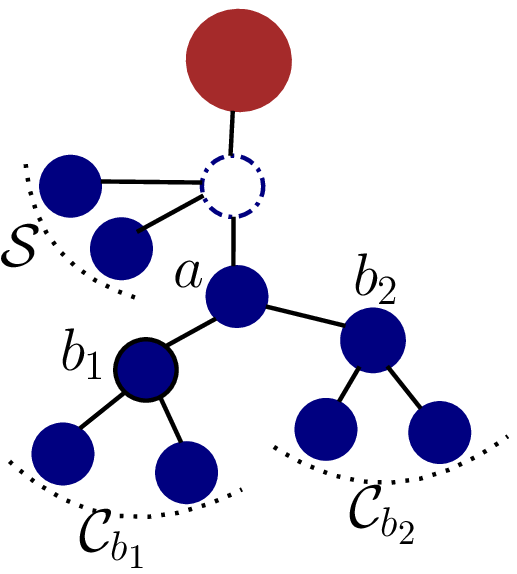}\label{fig:missing_2}}
\subfigure[]{\includegraphics[width=0.11\textwidth]{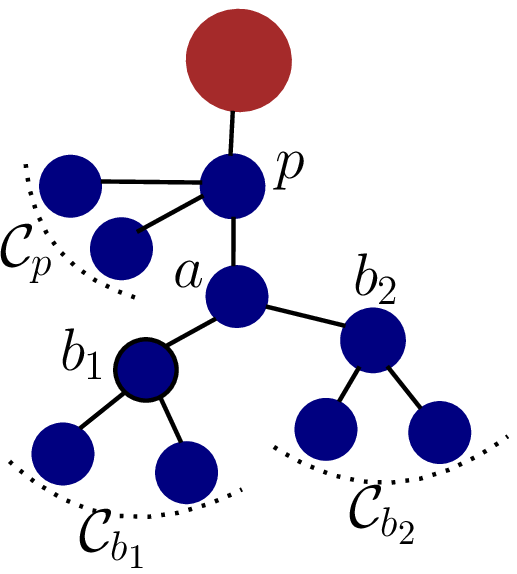}\label{fig:missing_3}}
\hspace*{\fill}
\squeezeup
\caption{Steps in Learning distribution grid in Fig.~\ref{fig:twohop} with hidden nodes $p,b_1,b_2$ (a) Spanning tree ${\mathcal T}_{\mathcal M}$ for observed nodes (b) Separation of children of node $a$ in $\hat{\mathcal T}$ into grandchildren and sibling sets with unknown parent nodes (c) Identifying parent node of $a$'s grandchildren, $a$'s parent unidentified (d) Identifying missing parent $p$ of node $a$.}
\label{fig:missing_example}\end{figure}

\textbf{Computational Complexity:} The complexity of Algorithm \ref{alg:3} can be computed similar to that of Algorithm \ref{alg:2} as the logic is similar. The primary difference in complexity arises due to separation between siblings and grandchildren of a node in Step (\ref{step_grandparent21}-\ref{step_grandparent22}) and identifying its missing children. This has complexity $O(|{\mathcal V}|^2)$. Iterating over all nodes, the complexity becomes $O(|{\mathcal V}|^3)$ in the worst case. This also dominates the overall complexity which is $O(|{\mathcal V}|^3)$. It is worth mentioning that the computational complexity results in the paper do not assume knowledge of tree depth, maximum node degree, or cardinality of permissible edge set. If they are known, the complexity can be reduced further.

\textbf{Extension to Multiple Trees:} Algorithms \ref{alg:1}, \ref{alg:2} and \ref{alg:3} can be extended to grids with multiple trees powered by different sub-stations. There we separate node groups for each tree before running the learning algorithms. This is possible as voltage magnitudes are measured relative to the root node, and hence voltages at two nodes $a$ and $b$ in distinct trees will be uncorrelated.

{\textbf{Correlated Injections:} Algorithms \ref{alg:1}, \ref{alg:2} and \ref{alg:3} and the theorems guaranteeing their correctness rely on the injection fluctuations being uncorrelated. Note that under correlated injections, $\Sigma_p, \Sigma_q, \Sigma_{pq}$ are not diagonal. Hence, the number of unknown variables (injection cross-correlations) increase in the case with missing nodes, and the current algorithms will not be able to estimate them. Exact injection estimation under correlated injections will be analyzed in future work. On the other hand, the correctness of the estimate of topology and injection variances (same node) by our algorithms under small correlated injections can be analyzed using perturbation theory. In particular the correlated covariance matrix $\Sigma_p$ (similarly for $q$ and $pq$) can be expressed as $\Sigma^{uc}_p+\Delta_p$, where $\Sigma^{uc}_p$ is the diagonal matrix of injection variances and $\Delta_p$ is the matrix of cross-covariances with zero in the diagonal. Consequently voltage covariances can be expressed as covariances under $\Sigma^{uc}_p$ and an error term. Thus voltage trends and equalities used in the learning algorithms can be satisfied up to a threshold for small injection correlations (small $\Delta_p$), and the correct topology can be learnt. We plan to study bounds on maximum injection correlation under which our algorithms are provably correct in future work.}

\textbf{Finite sample effect:} Empirically computed values of $\phi$ may differ from their true values and hence equalities and inequalities used in Algorithm \ref{alg:2} and \ref{alg:3} may only be satisfied approximately. As such we use user-defined tolerances ($\tau_1$ in Eq.~(\ref{first}), $\tau_2$ in Eq.~(\ref{third}), and $\tau_3$ in Theorem \ref{Theoreminequality}) to establish if the desired equalities/inequalities are true. To reduce the effect of varying injection covariances, we use thresholds based on the relative values (expressed as ratio) in the equality relations to determine their correctness.
%We select thresholds such that at extremely high samples, the output of the algorithms are correct. We keep the thresholds as constant while reducing the sample sizes and discuss the performance of the learning algorithms. In particular, historical data and test cases may be used to iterate and identify necessary tolerance values.
In the next section, we discuss the performance of our learning algorithms in test distribution networks, notably on voltage samples generated by non-linear AC power flows.

\section{Experiments}
\label{sec:experiments}
\begin{figure}[bt]
\centering
%\hspace*{\fill}
\subfigure[]{\includegraphics[width=0.14\textwidth]{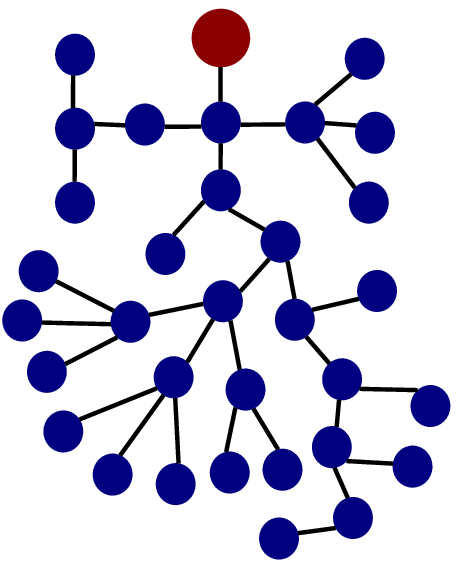}\label{fig:case33_algo1}}\hfill
\subfigure[]{\includegraphics[width=0.14\textwidth]{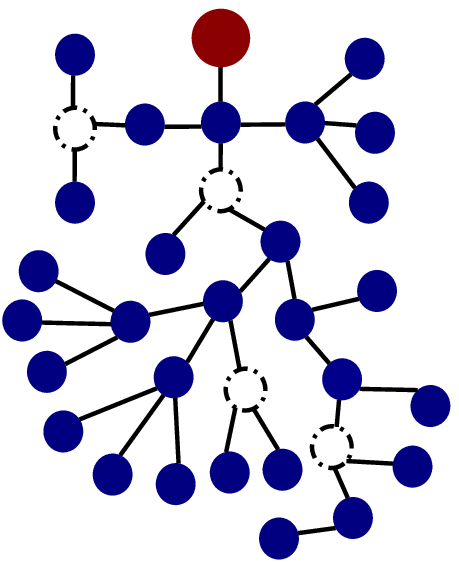}\label{fig:case33_algo2}}\hfill
\subfigure[]{\includegraphics[width=0.14\textwidth]{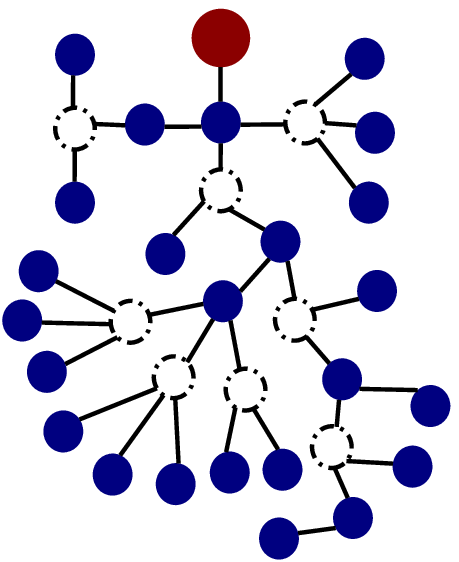}\label{fig:case33_algo3}}
\squeezeup
\hspace*{\fill}
\caption{Modified $33$-bus test case \cite{matpower} with observed nodes (solid blue), missing nodes (uncolored circles) for Algorithms \ref{alg:1} (a), \ref{alg:2} (b), and \ref{alg:3} (c)}\label{fig:case33}
\end{figure}
\subsection{Comparison of LC-PF and AC-PF}\label{sec:compare}
We demonstrate the accuracy of LC-PF Eqs.~(\ref{LC_PF}) for the modified $33$-bus test case \cite{matpower} in Fig.~\ref{fig:case33_algo1}. The modification is done to ensure hidden nodes in subsequent simulations have degree greater than two. Fig.~\ref{fig:approx} compares voltage magnitudes at non-substation buses computed by LC-PF with AC-PF solver in Matpower \cite{matpower} for two different variances in nodal injections relative to mean injection (range of $10^{-2}$ and $10^{-3}$). The voltages are measured relative to the per unit (p.u.) value at the reference bus. Note that the values are close. Hence theoretical algorithms proven for linearized power flow are able to perform estimation tasks with Matpower generated voltage measurements as presented next.
\begin{figure}[!hbt]
\centering
\includegraphics[width=0.38\textwidth]{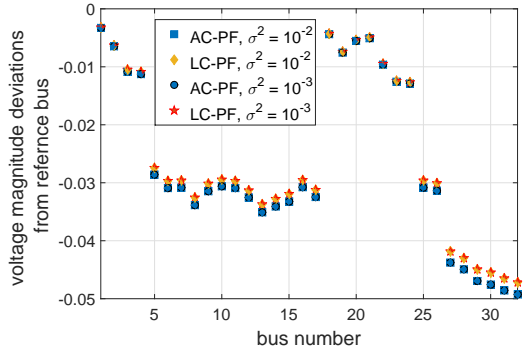}
\vspace{-.10cm}
\caption{Bus Voltage magnitudes (p.u.) by AC-PF (Matpower) and LC-PF (\ref{LC_PF}) for two different ranges of injection variances}\label{fig:approx}
\end{figure}
\squeezeup\squeezeup\subsection{Algorithms' performance}
\begin{figure*}[!ht]
\centering\hfill
\subfigure[]{\includegraphics[width=.33\textwidth]{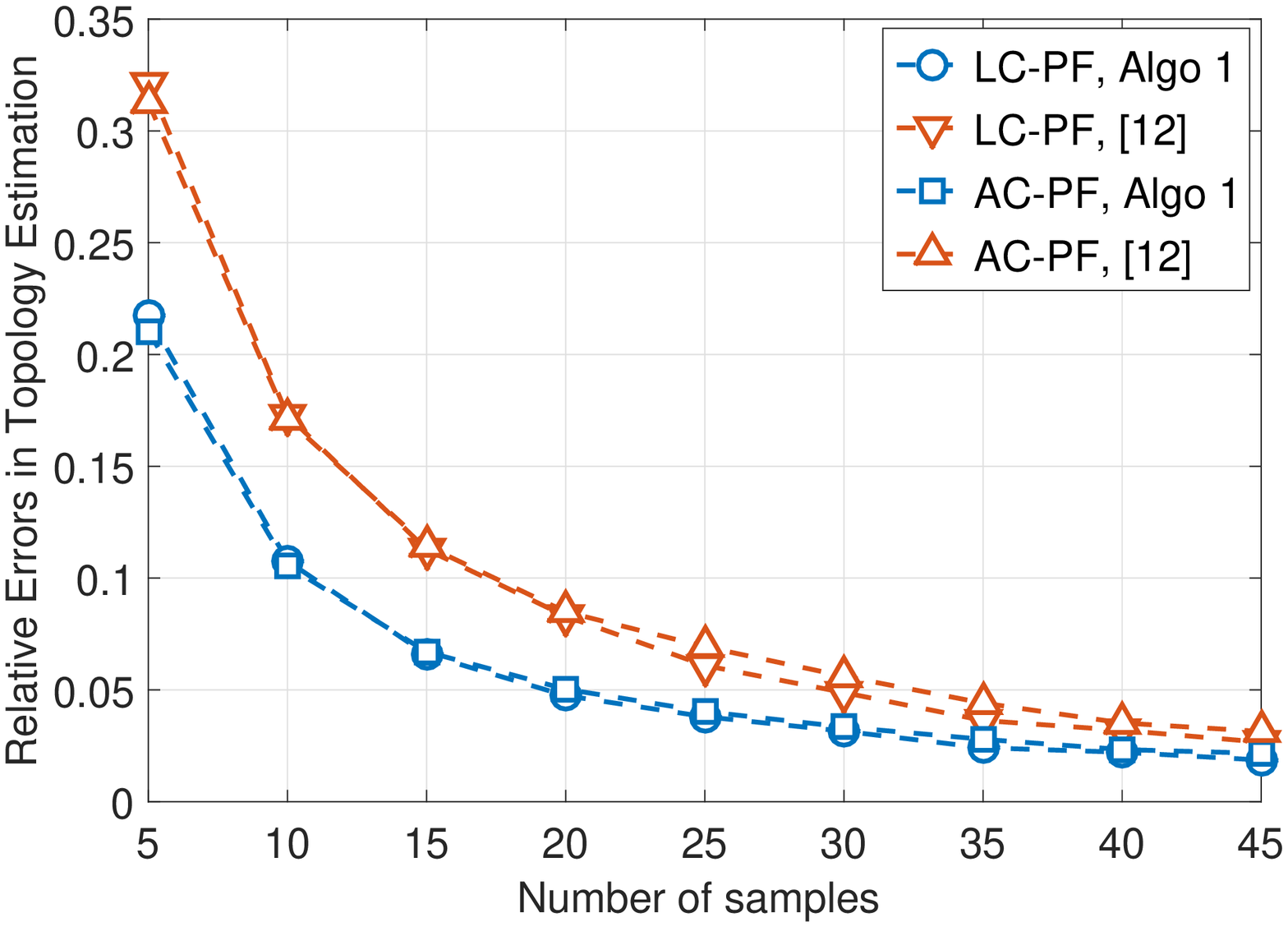}\label{fig:adj_algo1_tcns}}\hfill\subfigure[]{\includegraphics[width=.33\textwidth]{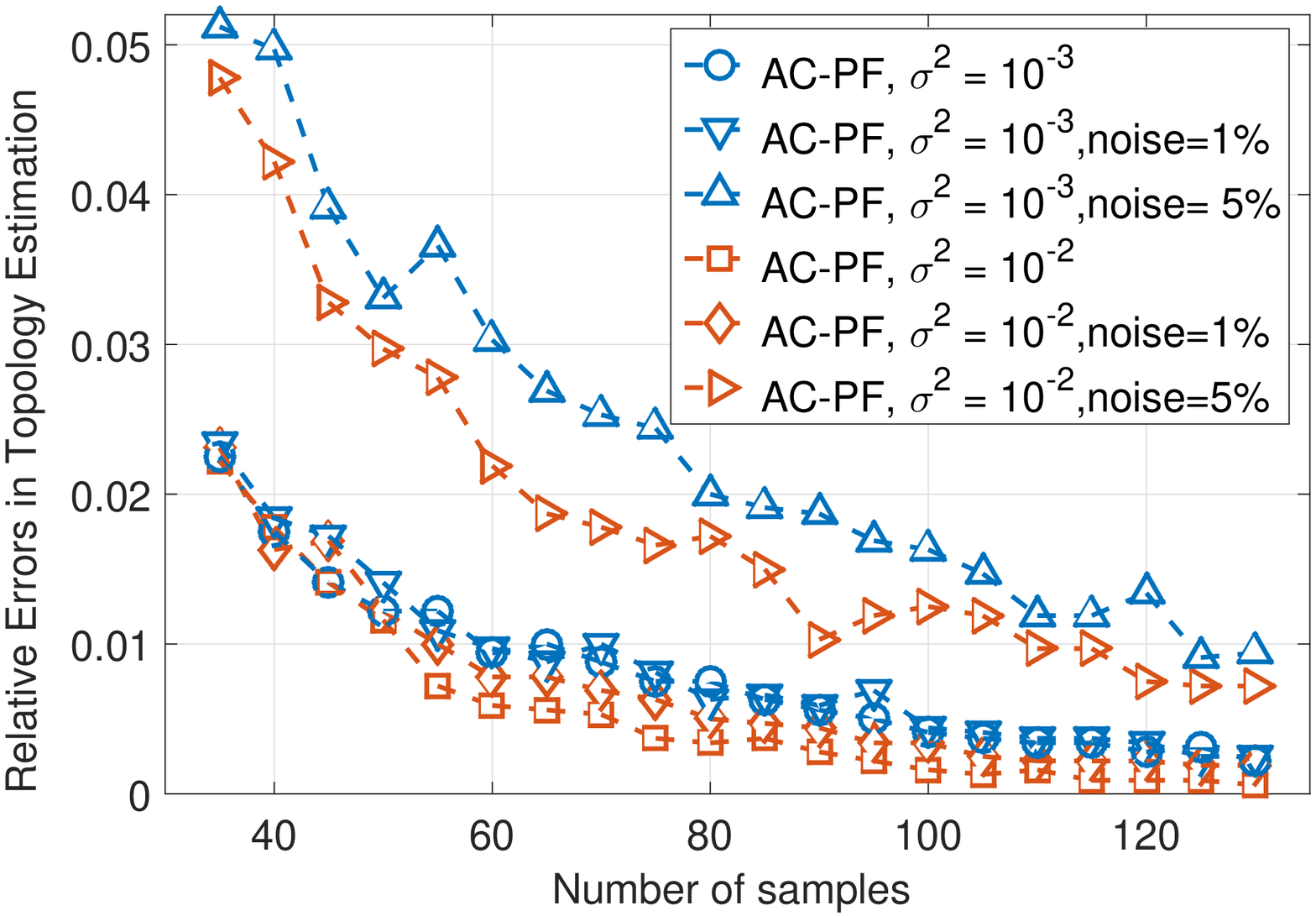}\label{fig:adj_algo1}}\hfill
\subfigure[]{\includegraphics[width=.33\textwidth]{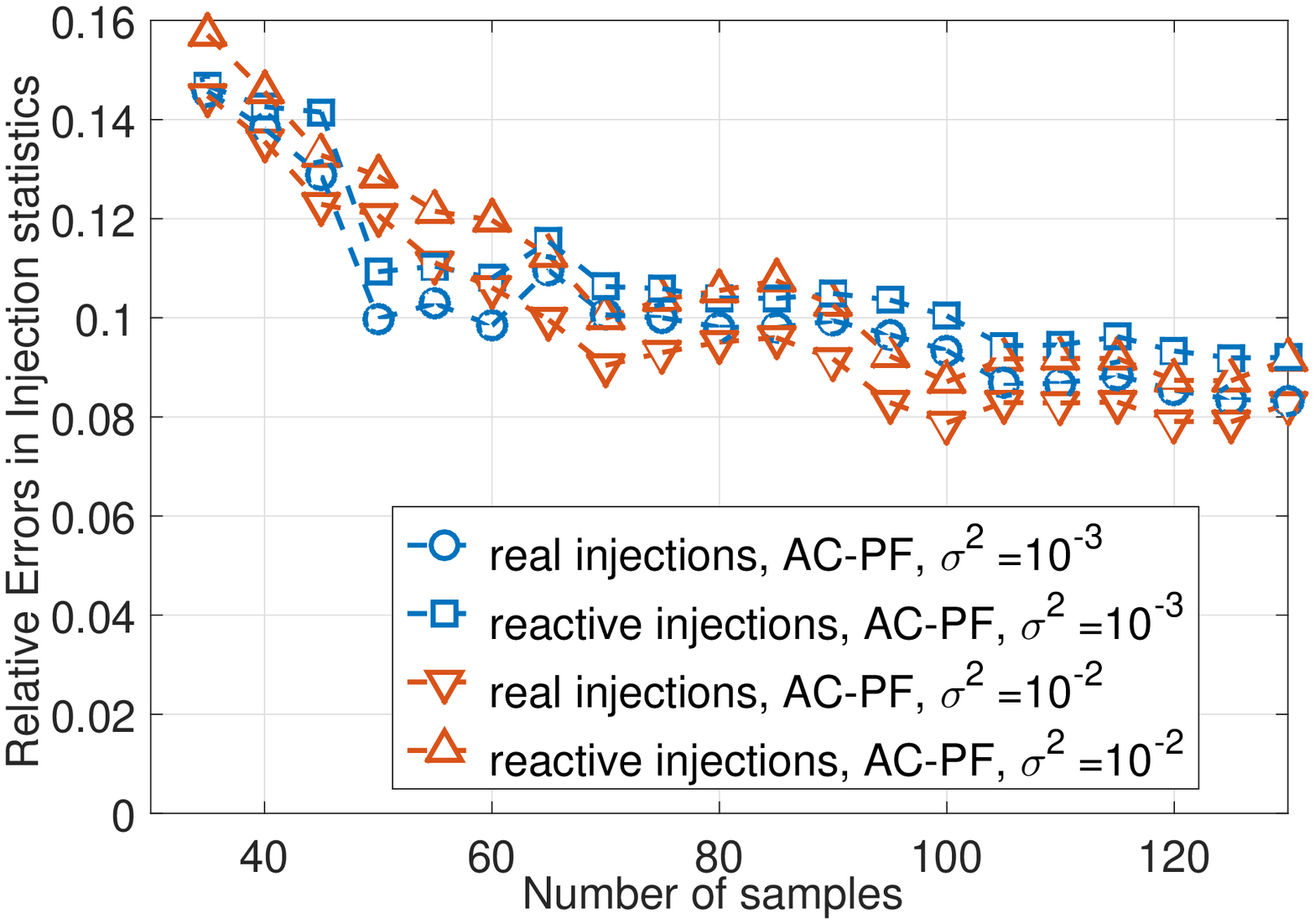}\label{fig:inj_algo1}}\hfill
\caption{(a) Comparison of topology estimation in Algorithm \ref{alg:1} with \cite{dekatcns} for grid in Fig.~\ref{fig:case33_algo1} with injection covariance of order $10^{-3}$ (b) Average relative errors in topology estimation for different noise levels, and (b) injection covariance estimation, v/s number of samples in Algorithm \ref{alg:1} for grid in Fig.~\ref{fig:case33_algo1}. Errors for two injection covariances are simulated.
}\label{fig:alg1}
\end{figure*}
We discuss the performance of our learning Algorithms $1,2,3$ in the test networks listed in Fig.~\ref{fig:case33}. To the operational $32$ edges, we add $50$ additional edges (at random) with similar impedances to create the input permissible edge set $\mathcal{E}_{full}$. To create the input set, we consider Gaussian active and reactive load fluctuations with random covariances selected relative to base loads. We consider two settings where the approximate order of covariances are taken as $10^{-3}$ and $10^{-2}$. The injections are uncorrelated across nodes and used to generate injection samples. These injections are then used to generate voltage samples with non-linear AC-PF solver Matpower \cite{matpower}. {We add independent Gaussian noise of fixed variance to the available voltage measurements to simulate noisy observations.} The set $\mathcal{E}_{full}$ along with the voltage samples and the injection statistics at the observed nodes are available as input to each algorithm. The observed nodes and hidden nodes are selected respecting Assumptions $2$ and $3$ as mentioned later. {Each plot presented in this section depicts average results over $1000$ independent realizations.}

{We first consider Algorithm \ref{alg:1} where voltages at all nodes are observed. We consider increasing number of noiseless LC-PF and AC-PF samples and present relative errors in topology estimation in Fig.~\ref{fig:adj_algo1_tcns}. The relative errors are computed as the differences between estimated and true edge sets measured relative to the number of total edges ($32$). Compared to the learning algorithm in \cite{dekatcns}, Algorithm \ref{alg:1} is not iterative and has better accuracy. Crucially, the errors under LC-PF and AC-PF are similar for Algorithm \ref{alg:1} due to sufficient accuracy of LC-PF samples as discussed in Section~\ref{sec:compare}. For the remaining simulations in the paper, we focus on AC-PF samples only.}

{In Fig.~\ref{fig:adj_algo1}, we present relative errors in topology estimation for different sample sizes and varying noise variances. We consider three noise variance settings: (a) noiseless, (b) $1\%$, and (c) $5\%$, relative to the measurement variance. Observe that the errors are insignificant beyond $60$ samples for both injection covariance settings considered, when noise variance is $1\%$ or less. For the $5\%$ noise case, the decay in error is slower and it takes around $120$ samples to reach the same level of accuracy.} The accuracy of estimated active and reactive injection statistics from noiseless AC-PF voltage samples is presented in Fig.~\ref{fig:inj_algo1}. The errors in injection statistics are measured relative to their true values and averaged over all nodes. Note that the estimate improves at higher samples as empirical moments are more accurate.

Next, we consider Algorithm \ref{alg:2} where missing nodes are separated by greater than two hops. We consider the setting in Fig.~\ref{fig:case33_algo2} with $4$ missing nodes and present results of topology and injection statistics estimation in Figs.~\ref{fig:adj_algo2} and \ref{fig:inj_algo2} respectively. Observe that the number of topology errors under both cases of nodal injection statistics reduce as the number of samples increase. However compared to Fig.~\ref{fig:adj_algo1} for no missing nodes, the number of samples needed is much higher. {Moreover the errors for noise variance of $1\%$ are much higher than than for the noiseless and $1\%$ noise setting. This is due to the fact that Algorithm \ref{alg:2} uses equality constraints (\ref{first},\ref{third}) to confirm true edges. At lower samples and higher noise, these constraints may not be satisfied up to the thresholds (pre-selected for the noiseless case), and hence errors are higher.}
\begin{figure}[!bh]
\centering%\hfill
\subfigure[]{\includegraphics[width=.35\textwidth]{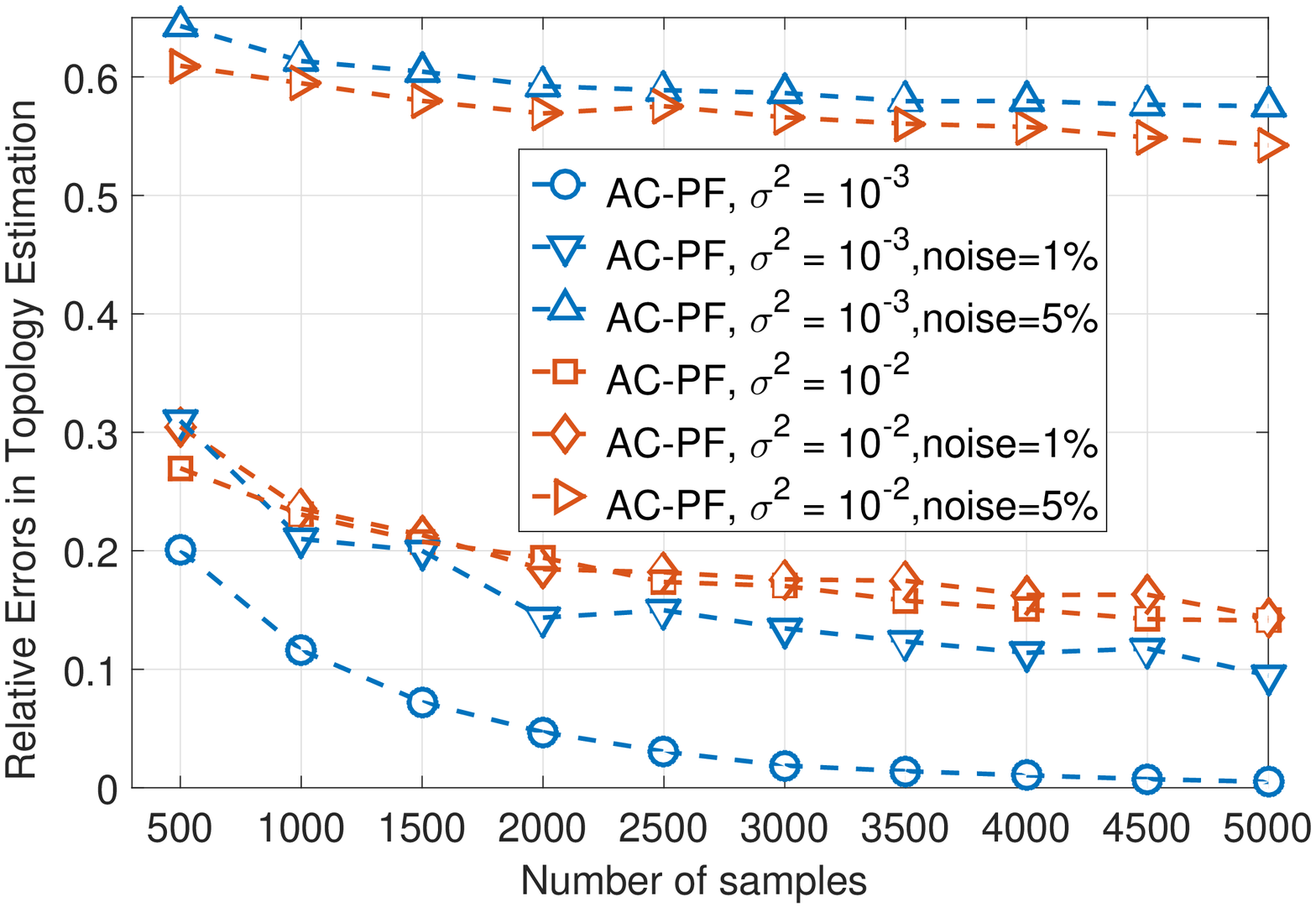}\label{fig:adj_algo2}}
\subfigure[]{\includegraphics[width=.37\textwidth]{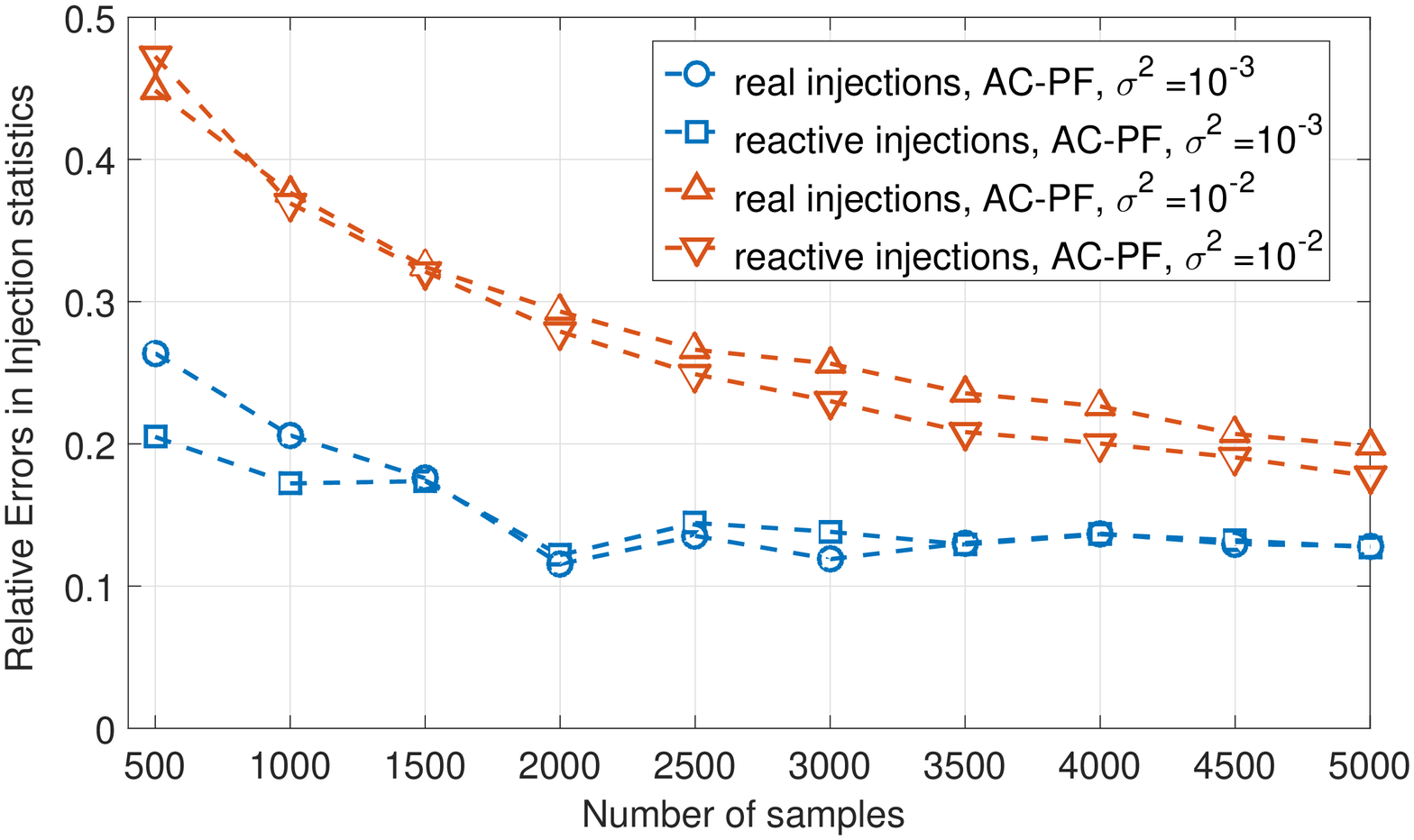}\label{fig:inj_algo2}}\squeezeup
\squeezeup\squeezeup\caption{ Average relative errors in topology estimation for different noise levels, and (b) injection covariance estimation, v/s number of samples in Algorithm \ref{alg:2} for grid in Fig.~\ref{fig:case33_algo2}}\label{fig:alg2}
\end{figure}
Finally, we consider Algorithm \ref{alg:3} that operates when hidden nodes are non-adjacent. We consider the setting in Fig.~\ref{fig:case33_algo3} with $8$ missing nodes ($4$ more than for Algorithm \ref{alg:2}). The performance for topology and injection statistics estimation are presented in Figs.~\ref{fig:adj_algo3} and \ref{fig:inj_algo3} respectively, for increasing voltage sample sizes. As before, the estimation errors decay with an increase in sample sizes for both injection covariance ranges selected. On expected lines, the performance of topology estimation worsens on increasing the noise level and decreasing the number of samples considered. Further note that the decay of errors in estimated injection statistics with increasing number of samples in each of the three algorithms is lower than that for topology estimation. This is not surprising as differences in estimated and true topologies are integer valued and depend on satisfaction of equality and inequality constraints within some threshold. On the other hand, errors in injection statistics are induced by real-valued differences with the true statistics that depend on empirical estimates and not just on the estimate of the true topology.
\begin{figure}[!bt]
\centering%\hfill
\subfigure[]{\includegraphics[width=.36\textwidth]{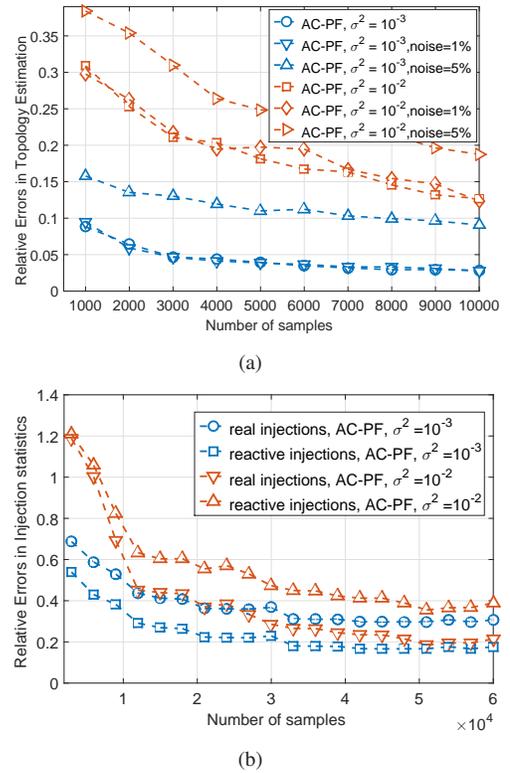}\label{fig:adj_algo3}}
\subfigure[]{\includegraphics[width=.36\textwidth]{33bus_inj_algo3_new.eps}\label{fig:inj_algo3}}
\caption{ Average relative errors in topology estimation for different noise levels, and (b) injection covariance estimation, v/s number of samples in Algorithm \ref{alg:3} for grid in  Fig.~\ref{fig:case33_algo3}}\label{fig:alg3}
\end{figure}

{\textbf{Effect of Threshold:} Note that unlike Algorithm \ref{alg:1}, Algorithms \ref{alg:2},\ref{alg:3} use thresholds $\tau_1$, $\tau_2$, and Algorithm \ref{alg:3} additionally uses $\tau_3$. The thresholds are picked to ensure correctness of output at large sample values ($4\times 10^4$ samples). To understand the impact of selected thresholds, we consider $5000$ noiseless voltage samples in Algorithms \ref{alg:2}, \ref{alg:3} for both injection covariances and vary each threshold relative to their pre-selected values, while fixing the others. It is clear from Figs.~\ref{fig:thres_algo2} and \ref{fig:thres_algo3} that Algorithms \ref{alg:2} and \ref{alg:3} are both more sensitive to $\tau_1$ (used for Eq.~(\ref{first})) than $\tau_2,\tau_3$ respectively. This can be explained as $\tau_1$ enables the preliminary determination of true parent-child edges that affects edge identification in follow-up steps. We postpone a theoretical study of correct threshold selection based on historical to future work.}
\begin{figure}[!bt]
\centering%\hfill
\subfigure[]{\includegraphics[width=.36\textwidth]{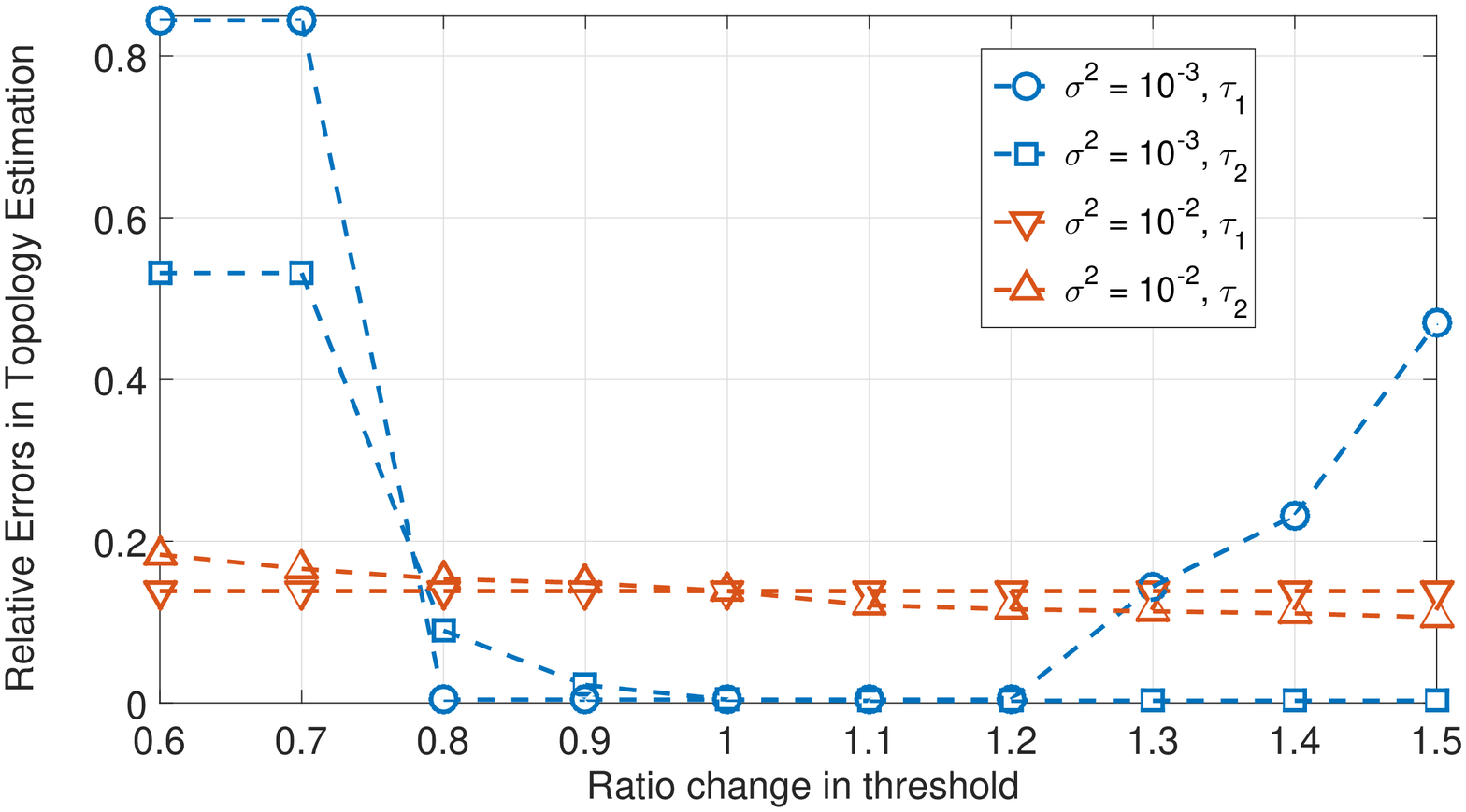}\label{fig:thres_algo2}}
\subfigure[]{\includegraphics[width=.36\textwidth]{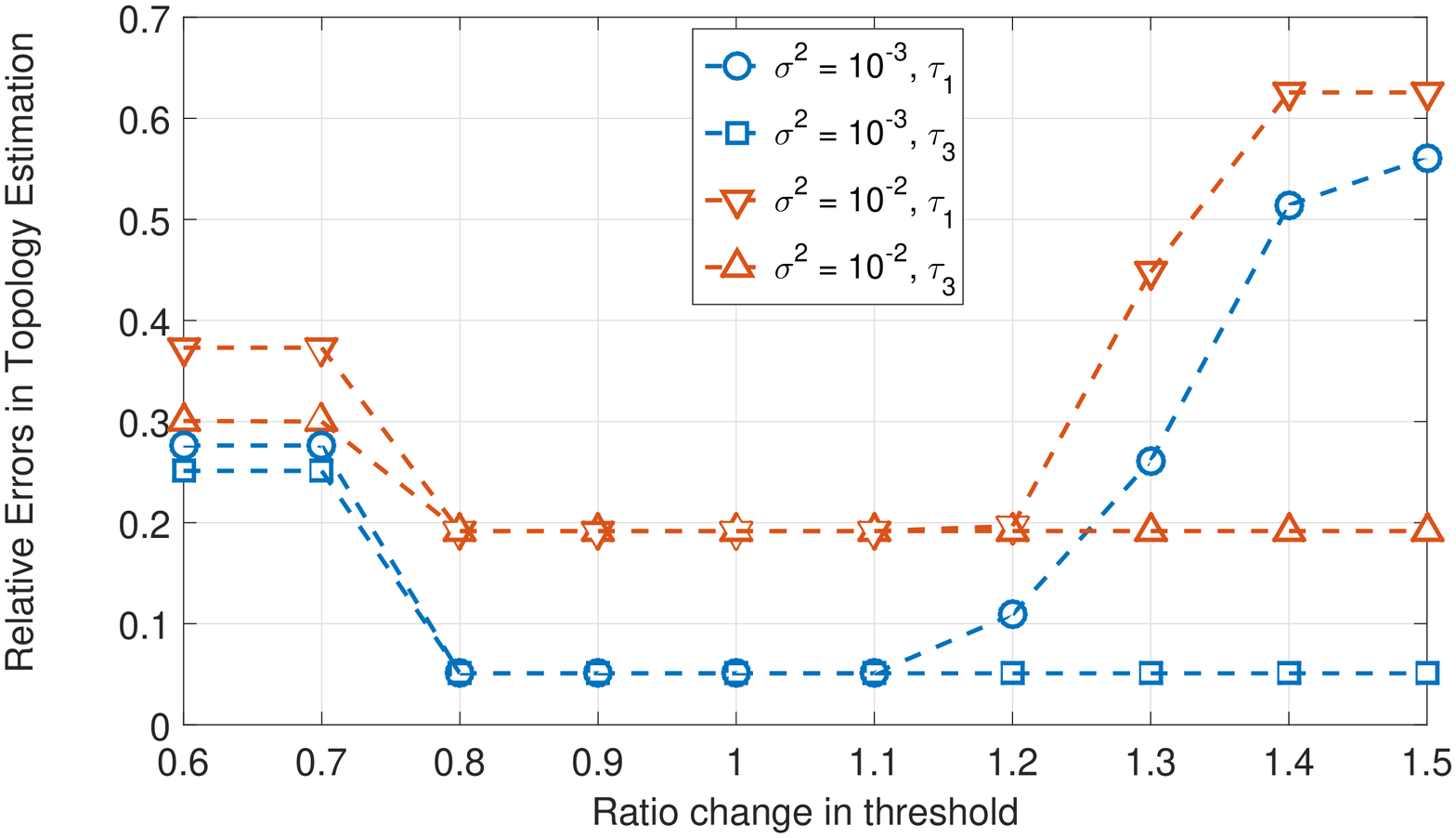}\label{fig:thres_algo3}}
\caption{Errors in topology estimation with relative change in thresholds (a) ($\tau_1, \tau_2$) in Algorithm \ref{alg:2}, and (b) ($\tau_1, \tau_3$) in Algorithm \ref{alg:3} for $5000$ voltage samples.}\label{fig:thres}
\end{figure}

\section{Conclusions}
\label{sec:conclusions}
This paper discusses algorithms for radial distribution grids to estimate the operational topology and injection statistics of missing nodes using voltage measurements and injection statistics at a subset of the grid nodes. We show that the learning algorithms provably learn the exact topology when all missing nodes are non-adjacent and have degree greater than two. Compared to previous work, the learning algorithms in this paper are able to handle a greater fraction of hidden nodes and require less information regarding them. Simulation results on test cases demonstrate the performance of the algorithms on realistic voltage samples generated by non-linear AC power flows.

In future we propose to extend the algorithms here to linearized multi-phase distribution networks \cite{dekathreephase,lowlinear}. {A formal understanding of the selection of thresholds and extension of the algorithm to cases with correlated injections are directions of future work.} The novel properties of voltage moments used in algorithm design may have applications in general network flow problems such as gas networks \cite{dekacdc}. We propose to analyze its relation to general graphical models.
\bibliography{sigproc,FIDVR,SmartGrid,voltage,trees}
\section{Appendix}
\subsection{Proof of Theorem \ref{Theoremcases}}
\label{sec:proof1}
\begin{figure}[!ht]
\centering
%\hspace*{\fill}
\subfigure[]{\includegraphics[width=0.13\textwidth]{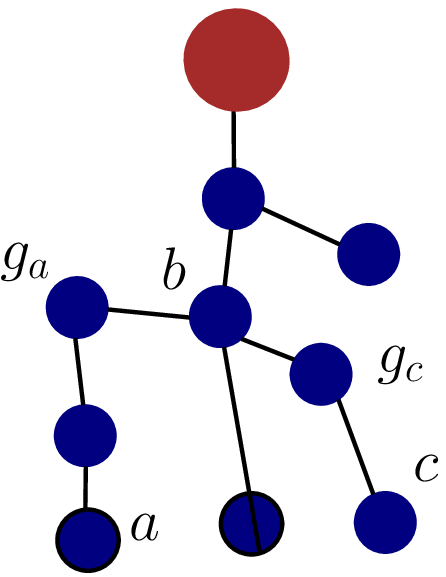}\label{fig:item2}}\hfill
\subfigure[]{\includegraphics[width=0.12\textwidth]{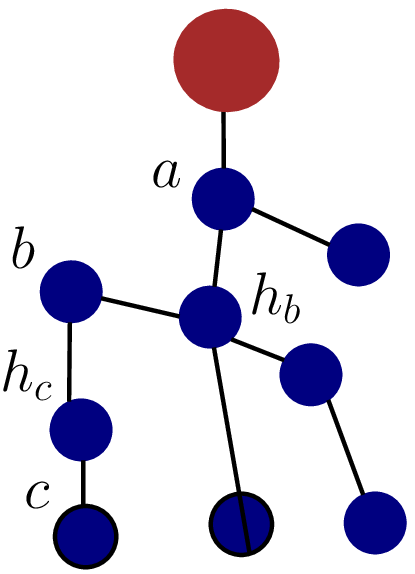}\label{fig:item3}}\hfill
\subfigure[]{\includegraphics[width=0.11\textwidth]{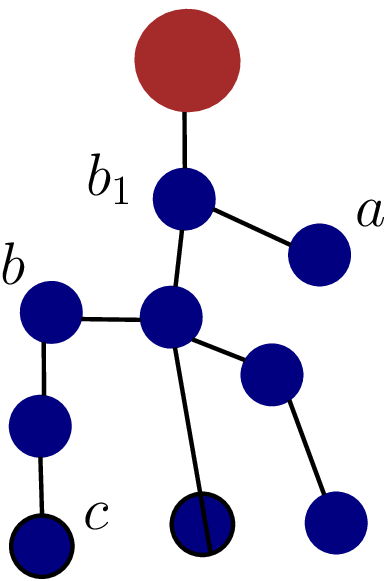}\label{fig:item4}}
\hspace*{\fill}
\caption{Distribution grid tree for proof of Theorem \ref{Theoremcases}.}
\end{figure}
\begin{proof}
As $b$ lies on the unique path from $a$ to $c$, we have $\mathcal{P}^c\bigcap \mathcal{P}^a \subseteq \mathcal{P}^b$. We first consider the case where $\mathcal{P}^c\bigcap \mathcal{P}^a = \mathcal{P}^b$ as shown in Fig.~\ref{fig:item2}. Here, both $a$ and $c$ are descendants of $b$ in $\mathcal{T}$. Let $g_a,g_c$ be $b$'s children on paths to $a$ and $c$ respectively. Clearly $\mathcal{D}^{g_a}$ and $\mathcal{D}^{g_c}$ are disjoint subsets of $\mathcal{D}^b$. Using Eq.~(\ref{Hrxinv}) and observing paths in this configuration, the following hold,
\begin{align}
&H^{-1}_{1/r}(a,d)= H^{-1}_{1/r}(b,d)\text{~~for $d \notin \mathcal{D}^{g_a}$}\label{cancel1}\\
&H^{-1}_{1/r}(c,d)= H^{-1}_{1/r}(b,d) \text{~~for $d \notin\mathcal{D}^{g_c}$}\label{cancel2}\\
&H^{-1}_{1/r}(a,d)= H^{-1}_{1/r}(c,d) \text{~~for $d \notin \mathcal{D}^{g_a}, \mathcal{D}^{g_c}$}\label{cancel3}
\end{align}
Similar results hold for $H^{-1}_{1/x}$ as well. Consider $\phi_{ab} + \phi_{bc}$ where the expansion of $\phi$ is given by Eq.~(\ref{usediff_1}). We denote the three additive terms on right side of Eq.~(\ref{usediff_1}) by $\phi^1,\phi^2,\phi^3$ for shorter expressions. We have
{\footnotesize
\begin{align}
 \phi^1_{ab} + \phi^1_{bc} = &\sum_{d \in \mathcal{D}^{g_a}}(H^{-1}_{1/r}(a,d)- H^{-1}_{1/r}(b,d))^2\Omega_p(d,d)\nonumber\\
 &~+\sum_{d \in \mathcal{D}^{g_c}}(H^{-1}_{1/r}(c,d)- H^{-1}_{1/r}(b,d))^2\Omega_p(d,d)\\
&= \sum_{d \in \mathcal{D}^{g_a}}(H^{-1}_{1/r}(a,d)- H^{-1}_{1/r}(c,d))^2\Omega_p(d,d)\nonumber\\
&~+\sum_{d \in \mathcal{D}^{g_c}}(H^{-1}_{1/r}(a,d)- H^{-1}_{1/r}(c,d))^2\Omega_p(d,d)=\phi^1_{ac}
\end{align}}
where the first equality follows from Eq.~(\ref{cancel1}), while the second equality uses Eqs.~(\ref{cancel2},\ref{cancel3}). Using the same logic for $\phi^2,\phi^3$ and adding them proves $\phi_{ab} + \phi_{bc}=\phi_{ac}$.

Next consider the case of $\mathcal{P}^c\bigcap \mathcal{P}^a \subset \mathcal{P}^b$. Here we first look at the configuration in Fig.~\ref{fig:item3} where $c$ is a descendant of $b$, which itself is a descendant of $a$. Let $h_b$ be the child of $a$ on path to $b$ and $h_c$ be the child of $b$ on path to $c$. As before, we consider $\phi = \phi^1+\phi^2+\phi^3$. Writing the expressions for $\phi^1$, we have

{\footnotesize
\begin{align}
 \phi^1_{ab} +\phi^1_{bc}&= ~\smashoperator[lr]{\sum_{d \in \mathcal{D}^{h_b}}}(H^{-1}_{1/r}(b,d)- H^{-1}_{1/r}(a,d))^2\Omega_p(d,d)\nonumber\\
 &~+ \sum_{d \in \mathcal{D}^{h_c}}(H^{-1}_{1/r}(c,d)- H^{-1}_{1/r}(b,d))^2\Omega_p(d,d)\\
 &=~\smashoperator[lr]{\sum_{d \in \mathcal{D}^{h_b}-\mathcal{D}^{h_c}}}(H^{-1}_{1/r}(b,d)- H^{-1}_{1/r}(a,d))^2\Omega_p(d,d)\nonumber\\
 ~+ \sum_{d \in \mathcal{D}^{h_c}}&((H^{-1}_{1/r}(b,d)- H^{-1}_{1/r}(a,d))^2+(H^{-1}_{1/r}(c,d)- H^{-1}_{1/r}(b,d))^2)\Omega_p(d,d)\\
 &\leq~\sum_{d \in \mathcal{D}^{h_b}}(H^{-1}_{1/r}(c,d)- H^{-1}_{1/r}(a,d))^2\Omega_p(d,d)= ~\phi^1_{ac}
\end{align}}
where we used the property that $H^{-1}_{1/r}(a,d)< H^{-1}_{1/r}(b,d)\leq H^{-1}_{1/r}(c,d)<0$ for $d \in \mathcal{D}^{h_b}$. As similar inequalities hold for $\phi^2$ and $\phi^3$, we have $\phi_{ab} + \phi_{bc} < \phi_{ac}$ for the configuration in Fig.~\ref{fig:item3}. By symmetry it is easy to see that the inequality holds when positions of $a$ and $c$ are exchanged. For any other configuration with $\mathcal{P}^c\bigcap \mathcal{P}^a \subset \mathcal{P}^b$ (Example in Fig.~\ref{fig:item4})), one can find an intermediate node $b_1$ such that $\mathcal{P}^c\bigcap \mathcal{P}^a = \mathcal{P}^{b_1}$. Then using the above analysis,
\begin{align}
\phi_{ab}+\phi_{bc} =\phi_{ab_1}+\phi_{bb_1}+\phi_{bc}< \phi_{ab_1}+\phi_{b_1c} = \phi_{ac}\nonumber
\end{align}
Thus it is true for all configurations with $\mathcal{P}^c\bigcap \mathcal{P}^a \subset \mathcal{P}^b$. Hence proved.
\end{proof}
\subsection{Proof of Theorem \ref{theorem:sample_complexity}}\label{sec:proofsample}
For zero mean injections, the voltage magnitudes (measured as deviations) are also zero mean. To prove the sample complexity result for Algorithm \ref{alg:1}, we first determine maximum empirical errors in $\phi_{ab} =\mathbb{E}[v_a -v_b]^2 $ that Algorithm \ref{alg:1} can tolerate. The following result holds for $\phi$.
\begin{theorem}\label{thm:order1}
Consider node $a$ in radial grid $\mathcal T$ with nodes $\mathcal V$ and edge set $\mathcal E$. $\forall b\neq a, (ab) \notin {\mathcal E}$, there exists some $c$ on path from $a$ to $b$ with $(ac) \in \mathcal E$ such that $\phi_{ab} \geq \phi_{ac} + k_1$ where \\
$k_1 = \min(r^2_{\min},x^2_{\min})\min\limits_{d \in {\mathcal V}}(\Omega_p(d,d)+\Omega_p(d,d)+2\Omega_{pq}(d,d))$. Here $r_{\min}, x_{\min}$ are the minimum values of resistance, reactance in grid $\mathcal T$.
\end{theorem}
\begin{proof}
Using Theorem \ref{Theoremcases}, it is clear that for any two non-neighbor nodes $a,b$ connected through $a$'s neighbor $c$, $\phi_{ab}-\phi_{ac} \geq\phi_{bb_1}$, where $b_1$ is $b$'s neighbor on path to $c$. Using Eq.~(\ref{first}) in Theorem \ref{Theoremcases2}, $\phi_{bb_1}$ is upper-bounded by $k_1$.
\end{proof}
Note that constant $k_1$ does not scale with the size of the network. For empirically computed $\hat{\phi}$'s, the correct output of Algorithm \ref{alg:1} follows from:
\begin{theorem}\label{thm:order2}
If $\forall a\neq b \in {\mathcal V}$, empirically computed $\hat{\phi}_{ab}$ satisfies $|\hat{\phi}_{ab}-\phi_{ab}|< k_1/2$, then Algorithm \ref{alg:1} outputs the correct topology.
\end{theorem}
\begin{proof}
Consider any $a$, its neighbor $b$ and node $c$ connected to $a$ via $b$. For correct topology estimation, we need $\hat{\phi}_{ab} < \hat{\phi}_{ac}$. This holds as
\begin{align}
 \hat{\phi}_{ac}- \hat{\phi}_{ab} &= (\hat{\phi}_{ac} - \phi_{ac}) + (\phi_{ab}- \hat{\phi}_{ab}) + (\phi_{ac}-\phi_{ab})\nonumber\\
 &> -k_1/2 - k_1/2 +k_1 > 0
\end{align}
Here, we use $x >-|x|$ and Theorem \ref{thm:order1}.\end{proof}
Using $\phi_{ab}=\Omega_v(a,a)+\Omega_v(b,b)-\Omega_v(a,b)-\Omega_v(b,a)$, errors in empirical estimates $\hat{\phi}_{ab}$ can be related to empirical voltage magnitude covariance $\hat{\Omega}_v$ as follows
\begin{align}
 \mathbb{P}[|\hat{\phi}_{ab}-\phi_{ab}|>\frac{k_1}{2}] \leq \mathbb{P}\left[\smashoperator[r]{\bigcup\limits_{i,j \in\{a,b\}}}|\hat{\Omega}_{v}(i,j)-\Omega(i,j)|>\frac{k_1}{8}\right]\label{unionprob}\end{align}
To complete the proof of Theorem \ref{theorem:sample_complexity}, we determine the number of samples necessary to ensure Theorem \ref{thm:order2} holds with high probability. We list the following result from \cite{ravikumar2011high, anandkumar2012high}.
\begin{theorem}\label{thm:gauss}
For a p-dimensional zero-mean Gaussian random vector $\textbf{X} = [X_1, ... X_p]$, empirical covariance from $n$ samples satisfies
$$ \mathbb{P}\left[|\hat{\Omega}(i,j)-\Omega(i,j)| >\epsilon \right]\leq 4\exp{\left[-\frac{n\epsilon^2}{3200\max_i{\Omega^2(i,i)}}\right]}.$$
\end{theorem}
As we consider voltage magnitudes as random variables, we determine the maximum value of diagonal of $\Omega_v$ first.
\begin{theorem}\label{thm:order3}
Under LC-PF model in grid $\mathcal{T}$ with depth $d$ and node set $\mathcal V$ and zero-mean Gaussian injection deviations, each diagonal entry $\Omega_v(a,a)$ is upper bounded by $d^2|\mathcal{V}|k_2$, where
$k_2 = \max(r^2_{\max},x^2_{\max})\max\limits_{c \in {\mathcal V}}(\Omega_p(c,c)+\Omega_p(c,c)+2\Omega_{pq}(c,c))$. Here $r_{\max}, x_{\max}$ are the maximum values of resistance, reactance in grid $\mathcal T$.
\end{theorem}
\begin{proof}
Using Eq.~(\ref{volcovar1}), $\Omega_v(a,a)$ equals
$\sum_{c} {H^{-1}}^2_{1/r}(a,c)\Omega_{p}(c,c)$ $+{H^{-1}}^2_{1/x}(a,c)\Omega_q(c,c)$ $+2H^{-1}_{1/r}(a,c)H^{-1}_{1/x}(a,c)\Omega_{pq}(c,c)$. Using Eq.~(\ref{Hrxinv}) and definition of $r_{\max}, x_{\max}$ and depth $d$, we have $H^{-1}_{1/r}(a,c) \leq dr_{\max}$, $H^{-1}_{1/x}(a,c) \leq dx_{\max}$. Thus we have
\begin{align}
&\Omega_v(a,a)\nonumber\\
&\leq d^2\max(r^2_{\max},x^2_{\max}) \smashoperator[lr]{\sum_{c \in {\mathcal V}}}\Omega_p(c,c)+\Omega_p(c,c)+2\Omega_{pq}(c,c)\nonumber\\
&\leq d^2|\mathcal{V}|\max(r^2_{\max},x^2_{\max}) \smashoperator[lr]{\max_{c \in {\mathcal V}}}(\Omega_p(c,c)+\Omega_p(c,c)+2\Omega_{pq}(c,c)).\nonumber
\end{align}
\end{proof}
Note that constant $k_2$ is independent of the size of the network. Using the Union bound with Theorem \ref{thm:gauss}, empirical estimates of all voltage magnitude covariances $\hat{\Omega}_v$ from $n$ samples are bounded by
\begin{align}
 &\mathbb{P}\left[|\hat{\Omega}_v(a,b)-\Omega_v(a,b)| >\epsilon, \exists a,b \in {\mathcal V}\right]\leq \nonumber\\
 &4|\mathcal{V}|^2\exp{\left[-\frac{n\epsilon^2}{3200\max_a{\Omega^2(a,a)}}\right]}
 \leq 4|\mathcal{V}|^2\exp{\left[-\frac{n\epsilon^2}{3200d^4|\mathcal{V}|^2k^2_2}\right]}\nonumber
\end{align}
Consider $\epsilon =k_1/8$, and Eq.~(\ref{unionprob}), we have
\begin{align}
 &\mathbb{P}\left[|\hat{\phi}_{ab}-\phi_{ab}| >k_1/2, \exists a,b \in {\mathcal V}\right]
 \leq 4|\mathcal{V}|^2\exp{\left[-\frac{nk_1^2}{d^4|\mathcal{V}|^2k_3}\right]}\nonumber
\end{align}
where $k_3 = 3200k^2_28^2$ is independent of the grid size. Using Theorem \ref{thm:order2}, Algorithm \ref{alg:1} outputs the correct topology with probability at least $1-\eta$ if
\begin{align}
 \eta > 4|\mathcal{V}|^2\exp{\left[-\frac{nk_1^2}{d^4|\mathcal{V}|^2k_3}\right]}
 \Rightarrow n > Cd^4|\mathcal V|^2\log(|\mathcal V|/\eta)
\end{align}
where constant $C$ is independent of the grid size. This proves Theorem \ref{theorem:sample_complexity}.
\end{document}